\documentclass[11pt]{article}

\usepackage{enumitem}
\usepackage{comment}
\usepackage[english]{babel}
\usepackage{amsthm}
\usepackage[utf8]{inputenc}
\usepackage{xcolor}
\usepackage{amsmath}
\usepackage{physics}
\usepackage[margin=1in]{geometry}
\usepackage{titlesec}
\usepackage{natbib}
\usepackage{graphicx}
\usepackage{graphicx} 
\usepackage{float}

\newtheorem{corollary}{Corollary}

\newtheorem{definition}{Definition}
\newtheorem{proposition}{Proposition}

\newtheorem{example}{Example}

\usepackage{caption}
\usepackage{subcaption}

\usepackage{algorithm}
\usepackage{algpseudocode}

\usepackage{amsmath}

\DeclareMathOperator*{\argmin}{arg\,min}

\usepackage{xcolor}
\usepackage{hyperref}
\hypersetup{%
  colorlinks=false,
  linkbordercolor=red
}

\usepackage{natbib}
\usepackage{graphicx}
\providecommand{\keywords}[1]{\textbf{\textit{Keywords---}} #1}

\usepackage{lineno}
\title{
\textbf{  Effective  Experience Rating for Large Insurance Portfolios via Surrogate Modeling }
}

\author{
\textbf{Sebastián Calcetero Vanegas*, Andrei L. Badescu, X. Sheldon Lin} \\
Department of Statistical Sciences\\
  University of Toronto\\
  Toronto, Ontario \\
  *\texttt{sebastian.calcetero@mail.utoronto.ca} \\
  }
  
\begin{document}

\maketitle

\begin{abstract}

 Experience rating in insurance uses a Bayesian credibility model to upgrade the current premiums of a contract by taking into account policyholders' attributes and their claim history. Most data-driven models used for this task are mathematically intractable, and premiums must be obtained through numerical methods such as simulation via MCMC. However, these methods can be computationally expensive and even prohibitive for large portfolios when applied at the policyholder level. Additionally, these computations become ``black-box" procedures as there is no analytical expression showing how the claim history of policyholders is used to upgrade their premiums. To address these challenges, this paper proposes a surrogate modeling approach to inexpensively derive an analytical expression for computing the Bayesian premiums for any given model, approximately.  As a part of the methodology, the paper introduces a \emph{likelihood-based summary statistic} of the policyholder's claim history that serves as the main input of the surrogate model and that is sufficient for certain families of distribution, including the exponential dispersion family.   As a result, the computational burden of experience rating for large portfolios is reduced through the direct evaluation of such analytical expression, which can provide a transparent and interpretable way of computing Bayesian premiums.

\end{abstract}


\section{Introduction}
Credibility, experience rating, or more recently the so-called ``a posteriori" ratemaking are fundamental areas in actuarial science that enable actuaries to adjust premiums based on a policyholder's experience. Classical credibility theory, as formulated by \cite{Buhlmann2005course}, is based on the theory of Bayesian inference. The current understanding of a policyholder's risk behavior serves as the ``a priori" information, which when combined with the actual claims experience, results in an ``a posteriori" understanding of the policyholder's true risk behavior. This posterior knowledge can then be used to calculate upgraded premiums. The nature of the Bayesian credibility model depends on the goals of the modeling question and, more importantly, on the behavior of the data. Insurance data can be highly complex due to the size of the portfolio of policyholders and its high degree of heterogeneity. Consequently, several modeling approaches to calculate upgraded premiums have been explored in the actuarial literature.

As discussed by \cite{norberg2004credibility}, actuaries have primarily used simplistic assumptions under the classes of Bayesian conjugate-prior families, best linear approximations or the so-called Bonus-Malus Systems (BMS), e.g., \cite{asmussen2020chapter}, to approach the problem of experience rating.  Most commonly, they use the Bühlmann credibility formula and its variations (e.g., Bühlmann-Straub). The latter provides a non-parametric approach to credibility with a clear interpretation and inexpensive computations by approximating the predictive mean via a weighted average of the claim history's sample mean and the manual premium. However, this formula is extremely restrictive for experience rating applications because it only provides values for the predictive mean and not other quantities of interest to actuaries (e.g., variance, quantiles, probabilities, etc). Additionally, the resulting premium using Bühlmann's formula is not data-driven since it relies on strong assumptions such as linearity and lacks compatibility with different model structures (e.g., multivariate models, regression models, heavy-tail models). It is known that this expression is exact only for a limited modeling framework, e.g., \cite{jewell1974credible}, and thus provides poor performance when the true predictive mean is non-linear \cite{gomez2014suitable}.

In today's insurance landscape, granular level information is increasingly available, and more complex insurance products, such as telematics-based insurance, have been developed. As a result, it is necessary to use tailored credibility models that account for the complexity of the data sets while still meeting regulatory requirements and standards \cite{denuit2019multivariate}. Such complicated structures that may appear in general insurance cannot be easily addressed by the traditional conjugate prior credibility framework or simplistic analytical models. Furthermore, the inclusion of policyholder attributes is highly relevant, as noted by  \cite{ohlsson2008combining}, \cite{garrido2016generalized} \cite{xacur2018bayesian}, and  \cite{diao2019regression}, and complicates the nature of the credibility model. Therefore, it is imperative to use general Bayesian modeling approaches that are data-driven and can provide a more flexible framework from which several quantities of interest can be derived for ratemaking, claim reserving, and risk management applications.

The general Bayesian workflow has not been widely used in insurance applications due to the challenges of obtaining analytical expressions. Many of the Bayesian models that provide a reasonable fit to real insurance data sets are mathematically intractable, and deriving Bayesian premiums requires computationally expensive numerical approximations, such as simulations via Markov Chain Monte Carlo (MCMC) methods. See, e.g., \cite{xacur2018bayesian}, \cite{zhang2018bayesian}, \cite{ahn2021ordering}. Dealing with large and heterogeneous insurance portfolios is one of the significant challenges with these numerical methods, as a considerable number of simulations from the posterior predictive distribution of each policyholder are required. Additionally, if a large amount of experience is observed, these simulations must be repeatedly performed to update the premiums. Another issue with analytically intractable Bayesian models is the ``black-box" approach used to derive premiums. These quantities are obtained through numerical approximations and not with analytical expressions, which makes the upgrading premium system lack practical interpretation, thus becoming unappealing to practitioners as the results are not explainable to clients or regulators.

Our paper aims to facilitate and promote data-driven Bayesian credibility models for experience rating in insurance by proposing an effective approach to approximate the Bayesian pricing formula. There have been a limited number of studies in general insurance. For instance \cite{Cred1998} proposes to use B-Splines instead of the traditional linear credibility formula. \cite{landsman2002credibility} uses a second-order Bayes estimator to approximate the predictive mean under non-conjugate models. \cite{li2021dynamic} tackled this problem by approximating a dynamic Bayesian model via a mathematically tractable discrete Hidden Markov Model.  In contrast, in the statistical literature, there have been some works that attempt to simplify the Bayesian inference problem of obtaining accurate approximations of the posterior or predictive distribution.  Many of these ideas are based on the use of summary statistics (a dimension reduction function) and so simplify the inference problem as, in general, low dimensional quantities are much easier to work with. See e.g., \cite{sisson2018handbook}, and further references in the literature of Approximate Bayesian Computation (ABC).  Similar ideas can be used for experience rating to find a summary statistic of the policyholder's claim history that still provides the same information contained in that history. See, e.g.,  \cite{taylor1977abstract},  \cite{credstats1979}. 

In this paper, we propose a surrogate modeling approach to address the challenges of computational burden and lack of analytical expressions in Bayesian credibility models and to efficiently and transparently perform experience ratings on complex insurance datasets and large insurance portfolios. Our methodology derives an accurate analytical expression that approximates the Bayesian premiums resulting from any given Bayesian model.  As part of the methodology, we introduce a likelihood-based summary statistic of the claim history of a policyholder that provides insights on the likeliness of their claim experience according to the model.  Such a likelihood-based statistic is shown to have desirable properties along the spirit of sufficient statistics,  in the sense of performing a dimension reduction without losing information for the inference. In particular, such a statistic is sufficient for several parametric families, including the exponential dispersion family of distributions. Therefore, it can be used as the main input for the surrogate model, ensuring its accuracy in reproducing the pricing formula for non-tractable Bayesian models. The analytical expression linking the likelihood-based statistic to the Bayesian premium policy-wise allows for effective scaling of premium calculation to the portfolio level, particularly for large portfolios. In this scenario, we may compute true Bayesian premiums for a small percentage of policies, say 5\%, from the portfolio and use the analytical expression as a surrogate model to extrapolate the Bayesian premiums on the selected policies to the rest of the portfolio.  Furthermore, the analytical expression can be used for sensitivity analysis and provide some interpretation to the upgrading of premiums.

The paper is structured as follows. Section \ref{framework} provides an overview of the Bayesian framework for experience rating and highlights some of the primary challenges that will be addressed in the paper. Section \ref{surrogate} introduces the surrogate modeling approach to approximate any pricing formula and emphasizes the importance of sufficient summary statistics of the policyholder's claim history as inputs of such a model. In Section \ref{theindex}, we present the likelihood-based summary statistic, properties and interpretation for insurance ratemaking. Technical results are deferred to Appendix \ref{sufficiency}. Section \ref{estimation} explains the estimation of the surrogate model and how it simplifies the ratemaking process for large portfolios. Section \ref{simulation} describes a simulation study that demonstrates the accuracy of the pricing formula for different Bayesian models. Section \ref{realdata} presents a case study of experience rating using a real European automobile dataset, and finally, Section \ref{conclusion} concludes the paper.

\section{ The Bayesian credibility framework and issues}\label{framework}

The mathematical framework for classical  credibility, as introduced by \cite{Buhlmann2005course}, is based on the Bayesian hierarchical model:
\begin{align*} 
\mathbf{Y}_n  \vert \Theta=\theta & \sim_{iid} f(y \vert \theta, \mathcal{O}  ) \\
\Theta & \sim P( \theta  )
\end{align*} 
\noindent
where $\mathbf{Y}_n = (Y_1, \ldots, Y_n)$ and each  $Y_j, ~ i=1, \ldots, n$ are the claim history (either claim size, frequency or other variables of interest) in period $i$ of a policyholder, $f(y \vert \theta, \mathcal{O}  )$ is the density function of the \emph{model distribution}, $\mathcal{O}$ is the set of parameters that contains other information available, e.g., policyholder attributes,   $\Theta$ is the \emph{latent variable} that represents the unobservable risk of each policyholder and has a prior distribution $P(\theta)$. The variables  $Y_j$ and $\Theta$ can be either univariate or multivariate, however, we shall keep the notation as if they are univariate for reading convenience. 

This formulation includes more general hierarchical models with several layers commonly used in credibility, e.g., \cite{frees2008hierarchical}, \cite{crevecoeur2022hierarchical}. For instance,  a Bayesian model can have two (or more) layers of latent variables as follows:
\begin{align*} 
\mathbf{Y}_n  \vert \Theta_1=\theta_1 & \sim_{iid} f(y \vert \theta_1, \mathcal{O}  ) \\
\Theta_1 \vert \Theta_2=\theta_2 & \sim P( \theta_1 \vert  \theta_2  )\\
 \Theta_2 & \sim P( \theta_2  )
\end{align*} 
\noindent
Do note that this is still embedded in the general formulation above by considering  $\Theta = (\Theta_1, \Theta_2) \sim P(\theta)= P(\theta_1,\theta_2)= P(\theta_2)P( \theta_1 \vert  \theta_2  )$. 

Similarly, the classical Bayesian setup also includes the case of regression-type credibility models such as panel data models widely used in experience rating in insurance. See, e.g.,  \cite{tzougas2021multivariate}, \cite{denuit2007actuarial}, \cite{bermudez2017posteriori}, \cite{desjardins2023hierarchical}, \cite{boucher2009survey}. For instance, consider a Generalized Linear Mixed Model (GLMM) such as:
$$ \mathbf{Y}_n  \vert \Theta=\theta \sim_{iid} f(y \vert \theta, \mathcal{O}  )$$
$$ \eta(\theta; \langle \boldsymbol{x}, \boldsymbol{\beta} \rangle) =   \langle \boldsymbol{x}, \boldsymbol{\beta} \rangle +\varepsilon$$
$$ \varepsilon \sim P( \varepsilon   )$$

\noindent
where $f( \cdot )$  is the model distribution, usually a member of the exponential family, $\eta( \cdot )$ would be the so-called link function that links the latent variable $\Theta$ to a regression on covariates $\boldsymbol{x}$ with coefficients $\boldsymbol{\beta}$, and  $\varepsilon$ is a random effect incorporated in the regression. Here, we are using implicitly the set of parameters as $\mathcal{O} =\langle \boldsymbol{x}, \boldsymbol{\beta} \rangle$. To ease the reading, we will drop  $\mathcal{O}$ in the notation, and only make explicit the dependence on covariates and their parameters $\mathcal{O} $ when necessary.

Given such a framework, the goal of experience rating consists of using the claim history $\mathbf{Y}_n$  to obtain a premium for the $(n+1)-th$ period. To do so, one first requires the predictive distribution of $Y_{n+1} \vert \mathbf{Y}_n$, which after simple manipulations can be computed as follows
\begin{equation}
f(Y_{n+1} \vert \mathbf{Y}_n) =\frac{ E_{\Theta} \left[   f(Y_{n+1} \vert \Theta ) \exp ( \ell (\mathbf{Y}_n \vert \Theta ) ) \right]}{ E_{\Theta} \left[    \exp( \ell (\mathbf{Y}_n \vert \Theta ) ) \right] }
\label{eqn1}
\end{equation}

\noindent where $E_{\Theta}$ is an expectation with respect to the prior distribution of $\Theta$ and $ \ell (\mathbf{Y}_n \vert \Theta )$ is the conditional log-likelihood of the observed experience $\mathbf{Y}_n$ given by
$$
\ell (\mathbf{Y}_n \vert \Theta ) = \sum_{j=1}^n \log( f(Y_j \vert \Theta ) ).
$$ 

The predictive distribution describes the risk behavior of a policyholder, and therefore it is the only quantity needed to compute any \emph{premium principle} or \emph{risk measure} that the actuary is considering. As defined by \cite{kaas2008premium}, a premium principle (or in general a risk measure)  is an operator $\Pi$ that assigns a given risk with a non-negative value, $\Pi(Y_{n+1} \vert \mathbf{Y}_n)$. Many premium principles have been studied in the actuarial literature as described in \cite{dickson_2005}  or \cite{radtke2016handbook}. The most well-known examples are premiums based on operators defined through conditional expectations of the form $ E( \pi(Y_{n+1}) \vert \mathbf{Y}_n )$ for some weight function  $\pi(y)$. Table \ref{table0} shows some examples.

\begin{table}[!htbp]
\center
\small
\begin{tabular}{ccc}
\hline
\hline
\textbf{Premium principle} & \textbf{Weight functions} & \textbf{Premium } $\Pi(Y_{n+1} \vert \mathbf{Y}_n)$\\ \hline
\emph{Net Premium} & $\pi(y) = y$ & $ E( \pi(Y_{n+1}) \vert \mathbf{Y}_n )$ \\ 
\emph{Expected value} & $\pi(y) = (1+\alpha) y$ & $ E( \pi(Y_{n+1}) \vert \mathbf{Y}_n )$ \\ 
\emph{Mean value / Utility function} & $\pi(y) = \mathcal{U}(Y)$ & $ E( \pi(Y_{n+1}) \vert \mathbf{Y}_n )$  \\ 
\emph{Variance /Standard deviation } & $\pi_1(y) = y , \pi_2(y) = y^2$ & $ E( \pi_1(Y_{n+1}) \vert \mathbf{Y}_n )+\alpha E( \pi_2(Y_{n+1}) \vert \mathbf{Y}_n ) $ \\ 
\emph{Exponential} & $\pi(y) = \exp(\alpha y)$  & $\log(  E( \pi(Y_{n+1}) \vert \mathbf{Y}_n ) )/ \alpha$ \\ 
\emph{Esscher} & $\pi_1(y) = y\exp(\alpha y), \pi_2(y) = \exp(\alpha y)$ & $ E( \pi_1(Y_{n+1}) \vert \mathbf{Y}_n )/E( \pi_2(Y_{n+1}) \vert \mathbf{Y}_n )$ \\ \hline
\hline
\end{tabular}
\caption{Example of Premium Principles}
\label{table0}
\end{table}

Although this Bayesian framework is in theory sound, its implementation is hard from a practical viewpoint, especially when one needs to account for the claim history and policyholder's attributes in large and heterogeneous insurance portfolios. Indeed,  the computation of the posterior predictive distribution in Equation (\ref{eqn1}) and any of the conditional expectations in Table \ref{table0}  could be a challenging task under most practical considerations. It is known that analytical expressions can only be obtained under certain choices of the model distribution and the prior, which apply mostly to a limited number of simple parametric models such as the conjugate families of distributions. Moreover, most of the data-driven setups that commonly provide satisfactory fits to real insurance data sets do not belong to this limited class of models, e.g., \cite{cheung2021bayesian}, \cite{czado2012mixed}, \cite{yau2003modelling}. As a result, performing Bayesian credibility is of much higher complexity due to the mathematical intractability of premium computations.

The most common approach used for the approximation of Bayesian premiums is via Markov Chain Monte Carlo (MCMC) methods along the same lines as Bayesian inference procedures. These methods enable the actuary to draw samples from the posterior and predictive distribution in Equation (\ref{eqn1})   without the need of having an analytical expression, and so the desired expectations in Table \ref{table0}  can be obtained as the sample average of the simulated quantities. This process can be computationally expensive as it requires a considerable amount of simulation of the desired quantities before obtaining reliable estimates, and the fact that it must be performed individually for each policyholder when considering the heterogeneity of policyholders' attributes on a regression-type model. That said,  when it comes to relatively large portfolios in non-life insurance, the repeated use of the MCMC process for each policyholder becomes almost prohibitive for practical applications because of the computational burden (\cite{ahn2021ordering}).

The issues mentioned earlier become more challenging when interpretability is at play. Ratemaking in insurance is a regulated activity and as a result, premiums must be transparent in terms of calculation so that policyholders are priced fairly without discrimination \cite{lindholm2022discrimination}. However, when numerical methods are used to calculate premiums (as in the case of Bayesian premiums via MCMC), the ratemaking process becomes a ``black box" due to the lack of an analytical expression that links the policyholder's attributes and their claim history to the resulting premium. Therefore, the ratemaking process becomes unexplainable and extremely unappealing for practitioners to find the true Bayesian premium. 

Actuaries have approached this issue with the so-called credibility framework, as mostly described by \cite{Buhlmann2005course}. To introduce the main idea, note that any Bayesian formula under any premium principle can be generally expressed as:
\begin{equation}
\Pi(Y_{n+1} \vert \mathbf{Y}_n) = G_{\Pi}( \mathbf{Y}_n, n, \mathcal{O}  )
\label{eqn3}
\end{equation}
where $G_{\Pi} (\cdot)$ is the theoretical or true functional form that links the claim history of the policyholder $\mathbf{Y}_n$ and the set of model parameters $\mathcal{O}$ to the Bayesian premium, under the premium principle $\Pi$. Essentially, the function $G_{\Pi}  (\cdot)$ is the integral of the premium principal with respect to the posterior predictive distribution, which, in turn, is the integration of the posterior distribution, as shown in Equation (\ref{eqn1}). The functional form of $G_{\Pi} (\cdot)$ entirely depends on the premium principle and the underlying Bayesian model, and as we mentioned above, it likely lacks an analytical expression. 

Along those lines, the traditional credibility approach approximates the Bayesian premium via a best linear approximation. For instance, when considering the net premium principle, \cite{Buhlmann2005course} shows that the best linear approximation of the Bayesian premium takes the form:
\begin{equation}
\hat{G}^{LI}_{\Pi}(\mathbf{Y}_n, n, \mathcal{O}  ) = Z_n\frac{\sum_{j=1}^n Y_j}{n}  + (1-Z_n) \Pi(Y_{n+1})
\label{linearcred}
\end{equation}

\noindent where $Z_n $ is the so-called credibility factor, and $\Pi(Y_{n+1})=E(Y_{n+1})$ is the so-called manual premium, which is the premium determined by the underlying Bayesian model without accounting for the claim history. Note that the information associated with the claim history of the policyholder is contained only in the summary statistic $\frac{\sum_{j=1}^n Y_j}{n} $. Here we use the notation $\hat{G}^{LI}$ as this is, in general, an approximation of the true functional form $G_{\Pi}$, rather than the true relationship. \cite{jewell1974credible} showed that under certain assumptions of a model distribution coming from the exponential family, conjugate prior distribution and some other regularity conditions, such a linear approximation is exact i.e. $\hat{G}^{LI}_{\Pi} \equiv G_{\Pi}$. This is known as the case of exact credibility in the literature. Do note that this is one of the very few cases where there is an explicit expression.

Unfortunately, the credibility approach is limited to three main aspects. The first one is with respect to the premium principles on which it operates. Briefly, the credibility approach has mainly been used to approximate the net premium principle only, rather than general quantities. Although some studies have been performed to extend the linear approximation to other premium principles, e.g., \cite{de1979numerical}, \cite{gomez2008generalization}, \cite{najafabadi2010new}, \cite{xie2018extension} or \cite{zhang4170554experience}; this must be achieved on a case-by-case basis rather than in a general framework due to the different nature of premium principles. The second aspect is the difficulty of the computations. Even though the linear expression simplifies the calculation of analytical expressions, the credibility factors still require analytical calculations which might be as hard to obtain as the true Bayesian premium itself depending on the complexity of the underlying model. Lastly, the third main issue is the linearity of the formula. Premium principles often lead to non-linear expression, and so the assumption of a linear approximation is by nature suboptimal and potentially inaccurate.

Along those lines, the use of the traditional linear credibility formula does not provide an effective approach to the calculation of premiums based on the claim history of policyholders under data-driven models. Therefore it is essential to be able to find a solution that provides accurate premium estimates, reduces the computational burden and is also interpretable. In the following sections, we propose a methodology that aims to address these challenges.

\section{A surrogate modeling approach  }\label{surrogate}

In order to address the mathematical intractability and computational burden for obtaining Bayesian premiums for general Bayesian models in the context of large insurance portfolios, we propose using a tailor-made \emph{surrogate modeling} approach. Surrogate models can help reduce the effort of computing the output of a function multiple times for different input values, especially when computing such output is computationally expensive. In such cases, a surrogate model approximates the output function in a more efficient way than the original computational process via a so-called surrogate function that usually has a closed-form and is easy to evaluate. The surrogate model is trained by using some known input-output pairs that are calculated using the true generating mechanism and then used to extrapolate to new inputs to approximate their outputs. The training and extrapolation are assumed to be computationally inexpensive when compared to the real mechanism generating the outputs (see \cite{sobester2008engineering} for more on surrogate modeling). Some applications of such models have been explored in actuarial science in \cite{lin2020efficient} and \cite{lin2020fast}, and also in ratemaking in \cite{henckaerts2022stakes}. 

In our context, the target of the surrogate model is to approximate the general pricing function in Equation (\ref{eqn3}) in an inexpensive fashion. Once a fitted surrogate function, say $\hat{G}_\Pi (\cdot)$, is obtained, the experience rating process for a new policyholder is performed by direct evaluation of the surrogate function and the calculation of Bayesian premiums for large portfolios becomes now a tractable task. Moreover, the surrogate function provides an analytical expression that links the attributes and claim history of the policyholder with the resulting premium. Therefore, the ratemaking process becomes transparent as the function $\hat{G}_\Pi (\cdot)$ can be used for sensitivity analysis to interpret the upgrading of premiums and quantify the effect of the claim history and attributes. See, e.g., \cite{henckaerts2022stakes},  \cite{lin2020fast} and \cite{SoaSurrogate} for examples in insurance. 

The surrogate function approach offers a flexible and practical method for experience rating, that aligns with the same spirit that initially drove the adoption of linear credibility. In essence, one can think of the surrogate modeling approach as a means of extending the best linear approximation by accommodating potential non-linear relationships, and not limited to the net premium principle. The key distinction lies in the fact that the estimation is not carried out analytically, but rather numerically via computational methods and the data itself. Notably, the best linear approximation can be retrieved through a linear surrogate function with the appropriate parameterization, as we will elaborate on in Section \ref{interpolation}.

Unfortunately, creating a surrogate function for experience rating presents a non-trivial challenge, primarily due to two main factors: the heterogeneity of the portfolio of policyholders and the varying dimensionality inherent in policyholders' claim histories.

The first challenge, stemming from this heterogeneity, permeates various aspects of the problem. However, the primary concern for the surrogate model revolves around the multitude of inputs upon which the surrogate function will be applied. To train this surrogate model, we rely on a predetermined grid of points where both inputs and outputs are well-known. In our specific context, this involves selecting a cohort of representative policyholders and determining their respective premiums based on the underlying Bayesian model. Consequently, our primary concern is how to select these representative policyholders to ensure accurate extrapolation. Given the substantial heterogeneity within the portfolio, it is crucial that the chosen grid points adequately encompass the wide variation in inputs to prevent any biases during extrapolation.

To address this challenge, we employ techniques borrowed from population sampling  (see, for example, \cite{chambers2012introduction}), as discussed in greater detail in Section \ref{sampling}. By drawing from this representative group, we can select a modest proportion of policyholders (typically 1\% to 5\%) from the overall portfolio while still effectively capturing the overall diversity.

The second significant challenge, arising as a direct consequence of the portfolio's heterogeneity, pertains to the characterization of the function $G_{\Pi} (\cdot)$ by a large and variable dimensionality. This issue emerges because each value of the claim history $\mathbf{Y}_n$, along with all other policyholder attributes, represents an $n$-dimensional input for the function. Moreover, this dimensionality is not fixed, as the number of periods of observed claim history $n$ may differ from one policyholder to another, ranging from a single period to many. Consequently, there is considerable difficulty in defining a suitable surrogate function \cite{hou2022dimensionality}.

Several approaches have been explored in the statistical literature to address dimensionality issues, primarily focusing on reducing the complexity of high-dimensional feature spaces. These approaches are applicable in our context if we consider policyholders' attributes and claim history as features in the problem. In such cases, employing dimensional reduction techniques through non-linear transformations that summarize the data is common, as discussed in \cite{blum2013comparative}. Specifically, this entails identifying a lower-dimensional function of the variables of interest, which, in our case, includes the claim history and attributes of a policyholder. This lower-dimensional function should capture the same information contained in the entire data set but in a simpler representation. Several methods have been explored in data science and it is known from theoretical statistics that this task is optimally achieved via the so-called ``sufficient statistics," as elaborated in \cite{casella2021statistical}. 

Sufficient statistics play a key role in general statistical inference due to their optimal compression of information as well as mathematical properties.  In Bayesian inference, for instance, it can be shown that using a sufficient statistic, say $T(\mathbf{Y}_n)$, instead of the entire information $\mathbf{Y}_n$ as the basis for inference yields identical results. Mathematically, this means that the posterior predictive distribution of the next period would satisfy
\begin{equation}
f(Y_{n+1} \vert T(\mathbf{Y}_n) ) = f(Y_{n+1} \vert \mathbf{Y}_n), 
\label{posteriors}
\end{equation}

\noindent thus, for a posteriori inference, it is unnecessary to consider the entire claim history $\mathbf{Y}_n$, but only the information contained in the sufficient statistic $T(\mathbf{Y}_n)$. A formal proof of this statement can be found in  \cite{bernardo2009bayesian}, Section 4.5. 

As a result, the functional form of the function $G_{\Pi} (\cdot)$ in Equation (\ref{eqn3}) can be simplified as 
\begin{equation}
\Pi(Y_{n+1} \vert \mathbf{Y}_n)  = \tilde G_{\Pi}(T(\mathbf{Y}_n), n, \mathcal{O} ),
\label{eqn3.1}
\end{equation}

\noindent for a different, yet very similar, function $\tilde G_{\Pi}(\cdot)$ that depends on the claim history only as a function of such sufficient statistic, which is of a much lower and fixed dimension than $\mathbf{Y}_n$. As such, from a complexity point of view, this function is simpler and easier to learn by a surrogate model. Therefore, it is desirable to work with sufficient statistics instead of the whole information of the claim history directly.

To illustrate the concept, let's revisit the case related to exponential families as described in Equation (\ref{linearcred}). In this scenario, it becomes straightforward to identify a sufficient statistic and highlight the simplification of the pricing formula that ensues. To do so, note that for exponential dispersion families, the sufficient statistic is represented as  $T(\mathbf{Y}_n) = \frac{\sum_{j=1}^n Y_j}{n}$.  Consequently, the equivalent, albeit simplified, form of the function $G_{\Pi}(\cdot)$ becomes
$$
\tilde{G}_{\Pi}(T(\mathbf{Y}_n), n, \mathcal{O}  ) = Z_n T(\mathbf{Y}_n)  + (1-Z_n)\Pi(Y_{n+1}).
$$

In practice, sufficient statistics might not be known depending on the complexity of the underlying Bayesian model and the heterogeneity of the data, however, it is possible to have approximations for these quantities that are practical and lead to a similar result. Intuitively, one might argue that if a summary statistic $T(\mathbf{Y}_n)$ is not sufficient, yet it is able to summarize a large piece of the information properly, then the relationships in Equations (\ref{posteriors}) and (\ref{eqn3.1}) should hold approximately up to some extent.   This is indeed the fundamental idea behind the literature on Approximate Bayesian Computation (ABC)  to tackle issues in Bayesian inference, e.g., \cite{sisson2018handbook}, \cite{joyce2008approximately}, \cite{sunnaaker2013approximate} and \cite{fearnhead2012constructing} for further references.  In such a context, the inference is based on summary statistics rather than on the whole data, which is usually more convenient from a computational perspective.  Such summary statistics are not necessarily sufficient, yet still, provide as much information to have a similar inference as if the whole data was used. These statistics are much easier to compute than true sufficient statistics and provide reliable approximations for inference purposes (i.e. almost equality in Equation (\ref{posteriors}). 

Along those lines, similar ideas can be applied to the context of experience rating, and the construction of the surrogate function. Indeed, if we have available a summary statistic that provides a dimension reduction with not much loss of information, then one can aim to estimate the desired surrogate function, say $\hat{\tilde G}_{\Pi}(\cdot)$, based on the information of such summary statistic, and have a reliable approximation of the true premiums using Equation (\ref{eqn3.1}).  Unfortunately,  the choice of such summary statistics remains under-explored in the experience rating literature, despite some implicit reliance on the idea. See, e.g., \cite{kunsch1992robust}, where the sample median is used instead of the mean, or \cite{yan2022general}, which considers a general class of credibility formulas based on linear combinations of estimators of the mean. However, none of these approaches discuss the relevance of using such summary statistics, besides intuition, nor do they account for the heterogeneity of the portfolio. 

In general, the search for such summary statistics is an open problem in ABC and data science, e.g., \cite{sunnaaker2013approximate}. Presently, the existing proposals tend to emphasize computational aspects that may not align well with the objectives of experience rating in the insurance context. Given that the goal of the ratemaking process is to maintain transparency and comprehensibility, an ideal summary statistic should possess a clear and interpretable meaning when describing a policyholder's claim history risk behavior. Simultaneously, it should capture all the pertinent information contained within that history. Unfortunately, many of the currently proposed statistics fail to meet these criteria, rendering them less suitable for our specific applications.

In the next section, we will present a likelihood-based summary statistic for claim history, designed to be a valuable component in the surrogate model for experience rating. This summary statistic possesses several advantageous properties that are well-suited for ratemaking applications. These properties include ease of computation, interpretability, and efficient handling of insurance data.  The details of how to use this summary statistic and how to fit a surrogate model in the context of large portfolios are provided in Section \ref{estimation}.

\section{A likelihood-based summary statistic }\label{theindex}

In this section, we introduce a likelihood-based summary statistic that is carefully crafted to capture most of the information within a policyholder's claim history.  It can be interpreted as a measure associated with the likeliness of a policyholder exhibiting such claiming behavior. This summary statistic is tailor-made to align with the underlying Bayesian model, incorporates policyholder attributes, and can be employed by actuaries to enhance their understanding of the experience rating process. Here we focus on introducing the likelihood-based statistic and its properties in a more intuitive fashion and defer the more involved technical details to  Appendix \ref{sufficiency}.

As we discussed in the previous section,    we seek a summary statistic that serves as a sufficient statistic for the accurate calculation of Bayesian premiums.     In statistics, it is well-known that the likelihood, as a function of both the latent variable and the observed data, is always sufficient (Lemma 1 in \cite{mayoqualitative} or \cite{schweder2016confidence}). Therefore, it is natural to consider the likelihood function itself (or the log-likelihood) as a reasonable choice for the desired summary statistic. However, the likelihood cannot be used as a statistic since the value of the latent variable is not fixed. Despite this, the likelihood function provides enough motivation to define our candidate summary statistic as a particular section of the likelihood function, as follows:
\begin{definition}[  likelihood-based summary statistic ]
The  likelihood-based summary statistic  or just likelihood-based statistic for a policyholder is defined as the following:
\begin{align}
{\mathcal{L}} (\mathbf{Y}_n; \tilde{\theta}, \mathcal{O}  ) = \sum_{j=1}^n \log f(Y_j \vert \Theta = \tilde{\theta}, \mathcal{O})
\label{index}
\end{align}
for a value $\tilde{\theta} \in  \mathcal{R}_\Theta$ that is policyholder dependent, and that is determined in a data-driven fashion.
\end{definition}

\noindent

In this scenario, instead of considering all potential values of the latent variable within its domain $\mathcal{R}_\Theta$, we select a fraction only of specific values $\tilde{\theta}$, one for each policyholder in the portfolio. We then focus on the segment of the likelihood function that is defined exclusively through these selected values for the purpose of inference. The values $\tilde{\theta}$ are, in principle, thoughtfully determined so that the simplified representation of the likelihood function retains as much information from the whole likelihood function. As such, the values $\tilde{\theta}$ can be seen as tuning parameters aiming towards the best compression of information. 

As our primary objective is ratemaking rather than broad Bayesian inference, the selection of the values $\tilde{\theta}$ is obtained under a different, yet similar principle. To do so, we employ a data-driven approach where a flexible non-linear model is trained to establish a connection between each policyholder and their corresponding value $\tilde{\theta}$, to maximize the overall quality of approximation given by the surrogate model to the true Bayesian premiums. The latter is achieved by minimizing the square error between the true premiums and the ones provided by the surrogate model. The detail of this process is provided in Section \ref{interpolation}. That said, it is important to clarify that the $\tilde{\theta}$ values do not necessarily represent estimations of the realization of the latent variable $\Theta$ associated with each policyholder, but rather serve as calibration parameters seeking the optimal performance of the surrogate model.  Once the $\tilde{\theta}$ values are determined for each policyholder, the expression in Equation (\ref{index}) can be regarded as a summary statistic of the claim history of each individual. It is worth noting that such statistics can be applied to claim histories containing both frequency and severity data simultaneously without any additional considerations.

Our newly defined statistic exhibits desirable properties for experience rating purposes, as discussed below, and also satisfies desirable properties that guarantee the reliability of a surrogate model based on it. Indeed, in  Section \ref{sufficiency} it is shown that the posterior predictive distribution in Equation (\ref{eqn1}) can be approximated as closely as needed through a judicious selection of the values $\tilde \theta$. Moreover, it is also shown that the likelihood-based statistic serves as a sufficient statistic under certain scenarios, including cases involving members of the exponential dispersion families. 

It is worth mentioning that within the literature on Approximate Bayesian Computation, researchers have explored  other summary statistics (informally termed \emph{approximate sufficient statistics} \cite[p.~130]{sisson2018handbook}) based on the likelihood function.  For instance, \cite{alsing2018generalized} and subsequent studies utilize the ``score function" at a fixed value of $\tilde \theta \in \mathcal{R}_\Theta$ (which they term as the fiducial parameter) as a potential summary statistic. In this context, much like our approach, the fixed value of the latent variable serves as a tuning parameter with the objective of attaining the most informative representation through the summary statistic. However, it is essential to highlight that these alternative summary statistics, while valuable for computational purposes in obtaining comprehensive posterior distributions, lack a primary interpretive purpose. Moreover, they do not establish formal sufficiency properties through these heuristically motivated methods. Consequently, such approaches may not always be well-suited for applications in experience rating. \\

\noindent \textbf{Interpretation of the summary statistic: }  The resulting summary statistic is a log-likelihood of the claim experience of a policyholder, but under a probability measure defined by the associated value $\tilde{\theta}$. Thus it can be interpreted as a measure of the likeliness that a policyholder's claim history aligns with the expected behavior under parameters of the model $\mathcal{O}$ but under such probability measure.  Intuitively, if the claim behavior aligns with the model parameters, then there is no reason to change the current premium, but if these don't, then the premium must be revised. The details on how to adjust the premium based on the summary statistic will be explained in later sections.

\noindent \\
\textbf{Accounting for policyholder attributes: } The summary statistic is tailored to the model distribution, taking into account all the actuarial considerations (e.g., policyholder information, tail heaviness, dependence structure) that have been incorporated into the model. For instance, in models accounting for covariates, the set of parameters has the form $\mathcal{O} = \langle \boldsymbol{x}, \boldsymbol{\beta} \rangle $, where the regression function is introduced as another component in the summary statistic, independent of the value $\tilde \theta $. 

$$
{\mathcal{L}} (\mathbf{Y}_n; \mathcal{O} , \tilde{\theta}  ) = \sum_{j=1}^n \log f(Y_j \vert , \Theta=\tilde{\theta}, \langle \boldsymbol{x}, \boldsymbol{\beta} \rangle)
$$

As such, the summary statistics naturally incorporate the individual attributes and distributional behavior of each policyholder, in contrast to predefined fixed-form summary statistics such as the sample mean or median, which do not account for any of these specifications.

\noindent \\
\textbf{Additivity of the observed experience: } The summary statistic is an additive function of the policyholder's claim history, where each period of observed experience contributes a new term to the computation. Consequently, to update the premium based on  $n-1$ observations, $\mathbf{Y}_{n-1}$, we can use the already available summary statistic, $\mathcal{L} (\mathbf{Y}_{n-1}; \tilde{\theta} )$, and add the contribution of the new observation via the relationship $\mathcal{L} (\mathbf{Y}_n; \tilde{\theta} ) = \mathcal{L} (\mathbf{Y}_{n-1}; \tilde{\theta} ) + \log f(Y_{n} \vert \Theta =\tilde \theta)$. 

\noindent \\
\textbf{Sub-statistics for multivariate models: } In insurance applications where policies have multiple coverages, the variable of interest $Y_j$ is a vector of dimension $D$, which we shall denote as $Y_j = (Y_j^{(1)}, Y_j^{(2)}, \ldots, Y_j^{(D)} )$, with $Y_j^{(d)}$ being the $d$-component of the vector. In such cases, the distribution function $f(Y_j \vert \Theta)$ is a multivariate conditional distribution. Most of the credibility models in the literature are constructed under a conditional independence structure in which the different components of the vector are conditionally independent given the latent variable $\Theta$. See, e.g., \cite{frees2016multivariate}, \cite{tzougas2021multivariate}, \cite{denuit2019multivariate} or \cite{englund2008multivariate}. In such a case, we can write the model distribution as the product of the individual conditional distributions of the vector, that is $f(Y_j \vert \Theta) = \prod_{d=1}^D f(Y_j^{(d)}\vert \Theta)$. Therefore we could write the summary statistic for a policyholder as
$$
{\mathcal{L}} (\mathbf{Y}_n;  \tilde{\theta}  ) = \sum_{d=1}^D {\mathcal{L}}^{(d)} (\mathbf{Y}_n^{(d)}; \tilde{\theta}  )
$$
with 
\begin{align}
{\mathcal{L}}^{(d)} (\mathbf{Y}_n^{(d)};  \tilde{\theta}  ) = \sum_{j=1}^n \log f(Y_j^{(d)} \vert  \Theta=\tilde{\theta}) ~~~, d = 1 , \ldots, D.
\end{align}

 \noindent

Hence, the statistic can be broken down into $D$ sub-statistics, with each one corresponding to the claim history of a policyholder for the $d$-th component only. Each sub-statistics can then be interpreted individually in the same way as the global statistic. Thus, the likeliness of the whole claim history vector is obtained by adding up the individual sub-statistics, offering a lucid interpretation of the claim history for multivariate models.

\noindent \\
\textbf{Accounting for partially observed claim history: } 
In practice, claim information is often subject to policy limits and deductibles, which results in truncation and censoring, respectively. Therefore, the observed values $Y_j$ for these types of claims cannot be used directly in traditional summary statistics. However, the summary statistic can easily accommodate different types of partially observed claim information, making it a valuable tool for analyzing real-world datasets.

For instance, if the value $Y_j$ is a right-censored observation, its contribution to the statistic would be the term $\log( 1-F(Y_j \vert \Theta= \tilde \theta) )$ instead of $\log( f(Y_j \vert \Theta = \tilde \theta) )$, where $F$ is the cumulative distribution function associated with the model distribution $f$. The same concept can be applied to the case of truncation, or to other types of modifications.
An interesting case of censoring worth discussing is missing data, which can still be accommodated in the computation of the statistic. For example, if one of the $D$ components of $Y_j$ is entirely missing, then the contribution to the statistic can be seen as $\log( 1 ) = 0$. This is equivalent to not adding an observation to the associated sub-statistic. This property is useful in multivariate models when a policyholder has only an observed claim history in one line of business, and not in all.

\section{Estimation of the surrogate model}
\label{estimation}

In this section, we provide further details on how to fit the surrogate function $\hat{ \tilde G}_{\Pi}(\cdot)$ and the tuning of the values $\tilde{\theta}$ to develop a surrogate function approach. The key steps to perform this task are presented below, and the remainder of this section discusses how each of these steps is achieved.

\begin{enumerate}

\item Select a sub-portfolio of representative policyholders.

\item Compute Bayesian premiums using simulation/numerical schemes only on such sub-portfolio.

\item Estimate the parameter $\tilde \theta$ and the non-linear function $\hat{\tilde G}_{\Pi}( \cdot )$ using a method for interpolation, as chosen by the actuary.

\item Assess the accuracy of the fitted function and out-sample predictive behaviour.

\item Evaluate the fitted formula on the rest of the portfolio to obtain the premiums.

\end{enumerate}

\subsection{Selecting a representative sub-portfolio via population sampling}
\label{sampling}

The selection of a representative sub-portfolio is a crucial step in obtaining accurate results through the surrogate function approach. This sub-portfolio should be small enough to be computationally efficient, yet large enough to exhibit similar properties to the original portfolio, ensuring reliable extrapolation. For large portfolios, a sample size between 1\% and 10\% of the whole portfolio may suffice for this purpose. It is recommended to start with a proportion of 1\% and increase this number by 1 \% at a time if the results are unsatisfactory, as we further discussed in Section \ref{assesment}.

 Selecting representative policyholders is a well-studied problem in the statistical literature on population sampling,  and several methodologies have been developed for this purpose. See, e.g., \cite{chambers2012introduction}. Population sampling is a statistical technique used to gather representative information about a larger group, known as the population. Instead of studying the entire population, a smaller subset called a sample, is selected using methods according to a sampling design. Such a design considers equal or unequal inclusion probabilities, which determine the likelihood of each element being included in the sample.   
 
 Given the particular goal of extrapolation, it is critical for the sub-portfolio (sample) to exhibit the same characteristics as the total portfolio concerning all the inputs involved in computing the likelihood-based statistic, including the parameters of the model distributions (e.g., covariates in a regression), the number of periods with observed exposure $n$, and the claim history $\mathbf{Y}_n$. Therefore, we recommend using model-assisted sampling methods that can sample from large populations while accounting for several attributes of the sub-portfolio to match those of the entire portfolio.

In this paper, we use the \emph{cube method}, as described in \cite{tille2011ten}, which has been used in other applications of surrogate models in insurance, such as \cite{lin2020fast}.  The cube method allows for the selection of samples from a population of any size, taking into account both equal and unequal inclusion probabilities. The samples are chosen in a way that aims to achieve \emph{balance} across a specified set of attribute variables i.e. approximates equality or equilibrium in the distribution of the variables. Briefly, the cube method consists of two phases: a flight phase and a landing phase. The flight phase defines inclusion probabilities that guarantee the balance property on the selected set of attribute variables in the sub-portfolio, and the landing phase converts these probabilities to either zero or one using linear programming, resulting in an approximately balanced random sample. 

The algorithm has been implemented in several statistical packages, for instance, the function \texttt{samplecube} in the sampling package in \texttt{R}, \cite{tille2010teaching}. Therefore it is readily available for applications. We would like to remark that it is not computationally expensive as it is designed for large populations. The inputs are the data frame with the population of interest, including the attributes to keep balance in the sample, and starting point inclusion probability for each policyholder, which in our case is the sample size we want to sample from the population, e.g., 1\%. 

\subsection{Computing Bayesian premiums for the sub-portfolio }\label{importance}

The computation of Bayesian premiums can be achieved using numerical schemes based on either Markov chain Monte Carlo (MCMC) or quadrature methods. This problem is well-documented in the literature of computational Bayesian statistics. For further details, we refer the reader to works such as \cite{sisson2018handbook}. We emphasize that this step represents the bottleneck process that we aim to mitigate as much as possible. Therefore, we remind the reader to use the most efficient algorithm available to them. For the sake of completeness, we present here a tailored-made setup that is quite efficient for the computation of premiums defined through the expectation operators, as motivated in Section \ref{framework}.

We propose using a simulation approach via \emph{normalized importance sampling}, e.g.\cite{tokdar2010importance}, in which the \emph{proposal distribution} for the posterior distribution of the latent variables is chosen to be the prior distribution. Therefore, expectations from the posterior predictive distribution can be estimated as
\begin{equation}
E(\pi(Y_{n+1}) \vert  \mathbf{Y}_n) \approx \hat{E}(\pi(Y_{n+1}) \vert  \mathbf{Y}_n) := \frac{\sum_{k=1}^K   E \left[\pi(Y_{n+1}) \vert \theta_k \right] \exp ( \ell (\mathbf{Y}_n \vert \theta_k ) )}{\sum_{k=1}^K    \exp ( \ell (\mathbf{Y}_n \vert \theta_k ) )},
\end{equation}

\noindent where $\theta_k$ are iid samples drawn from the prior distribution $P(\theta)$. The value of $K$ is selected such that reasonable estimates are guaranteed.  Such an estimator possesses desirable properties including consistency, approximated normality and reduced variance, which makes it a desirable choice for our application. Moreover, unlike the traditional MCMC approach that generates samples from the posterior, this setup generates samples from the prior distribution, making it easy to perform. Additionally, the prior distribution $P(\theta)$ is the same for every policyholder in the portfolio. Therefore, a single $\theta_k$ drawn can be used across policyholder simulations, resulting in a generation of samples that scales up efficiently to large-size portfolios.

Finally, once these expectations are estimated, we are in a good position to find the associated Bayesian premium, according to the premium principle that is being considered (Table \ref{table0}). Therefore, we end up with a reliable approximation of such premiums, which we denote by $\hat{\Pi}(Y_{n+1} \vert \mathbf{Y}_n) $ from now on.

\subsection{Fitting of $\hat{ \tilde G}_{\Pi}(\cdot)$ and the values $\tilde \theta$ }
\label{interpolation}

To do this we use a non-parametric interpolation method in which the inputs are the set of parameters of the model, the likelihood-based statistic and the number of periods of claim history of each policyholder. The outputs are the Bayesian premiums for each policyholder in the sub-portfolio of representative policyholders.  It should be noted that the literature on the interpolation of general functions and surrogate modeling is extensive. See, e.g., \cite{mastroianni2008interpolation}. Thus, the following description serves as a guideline, rather than a step-by-step recipe to be followed. The key points of this process are the selection of a structure for the surrogate model, followed by its estimation via least squares.

\subsubsection*{Selecting a structure for the surrogate function}
In the literature, the most commonly used methods for fitting a surrogate function $\hat{\tilde G}_{\Pi} (\cdot)$ are regression via Gaussian processes, as described by \cite{sobester2008engineering}. However, other flexible non-linear models, such as neural networks and spline-based methods, may also be used.

While these methods provide accurate estimates, the surrogate function $\hat{\tilde G}_{\Pi} (\cdot)$ may have a non-clear structure of inputs, making it difficult to explain. To ensure transparency in the ratemaking process, it is necessary to select a particular structure for the surrogate function that allows for interpretation while also guaranteeing enough accuracy of the surrogate model. 

For example, one could choose a linear surrogate function in the same spirit as Equation (\ref{linearcred}), and recover the best linear approximation as a particular case of the surrogate function under a proper parameterization and a proper summary statistic. However, we emphasize that the simplistic linear structure might not be able to provide accurate approximations in general, especially when the premium principle in consideration is not linear in nature and the credibility model is complex.

To obtain meaningful interpretations of the ratemaking process while maintaining accuracy when interpolating, we illustrate the methodology using a surrogate function with a rating factor functional form as shown in Equation (\ref{eqn4}) below. However, do note that any other structure can be used.
\begin{equation}
\hat{ \tilde G}_\Pi(\mathcal{L} (\mathbf{Y}_n; \tilde{\theta} ),n, \mathcal{O}) =  \Pi(Y_{n+1}) \exp ( g( \mathcal{L} (\mathbf{Y}_n; \tilde{\theta} ),n))
\label{eqn4}
\end{equation}

\noindent where $\Pi(Y_{n+1})$ represents the manual premium; which is calculated using the premium principle, but under the \emph{prior} distribution (i.e., the associated premium the insurance charges if there is no claim history available), and $g(\mathcal{L}(\mathbf{Y}_n; \tilde{\theta}),n)$ is a non-linear function.

The expression given in Equation (\ref{eqn4}) resembles a process of rate upgrading, where $\Pi(Y_{n+1})$ represents the current premium the policyholder is charged,  and the term on the right-hand side, $\exp(g(\mathcal{L}(\mathbf{Y}_n; \tilde{\theta}),n))$, acts as a rating factor that adjusts the manual premium based on the claim history. Along those lines, the function $g(\mathcal{L}(\mathbf{Y}_n; \tilde{\theta}),n)$ describes how much of an adjustment should be made on the premium.  If the function $g(\mathcal{L}(\mathbf{Y}_n; \tilde{\theta}),n)$ equals 0, no adjustment is required; if $g(\mathcal{L}(\mathbf{Y}_n; \tilde{\theta}),n)>0$, the premium must be increased; and if $g(\mathcal{L}(\mathbf{Y}_n; \tilde{\theta}),n)<0$, the premium should be reduced. As a result, the ratemaking process acquires a more transparent view of the actuary, as the Bayesian premium can be interpreted and explained to both clients and regulators.  The function $g(.)$ can be non-linear to provide flexibility and should be monotonic on the likelihood-based statistic to ease the interpretations.  Section \ref{realdata} provides a clearer view of this interpretation in a case study.

\subsubsection*{Least squares estimation}

Consider a sub-portfolio of $M$ policyholders indexed by $i = 1, \ldots, M$. Let $\hat{\Pi}^p_i$ denote the Bayesian premium for the $i$-th policyholder obtained via simulation, i.e. the $\hat{\Pi}(Y_{n+1} \vert \mathbf{Y}_n)$ values. Since these values were obtained as the sample mean of a large number of simulations, we can assume that $\hat{\Pi}^p_i \approx \textrm{Normal}(\Pi_i^p, se^2_i)$, where $\Pi_i^p$ is the Bayesian premium $\Pi(Y_{n+1} \vert \mathbf{Y}_n)$ and $se_i^2$ is the standard error of the estimation. Let $\Pi_i$ denote the manual premium of the $i$-th policyholder, i.e. the premium $\Pi(Y_{n+1})$ without incorporating claim history, and let $\mathcal{L}_i(\mathbf{Y}_n,\tilde{\theta}_i)$ denote the likelihood-based at the value $\tilde{\theta}_i$ for the $i$-th policyholder in the portfolio.

Assuming the rating factor functional form, our objective is to fit the relationship in Equation (\ref{eqn4}), where we need to estimate the function $g(\cdot)$ that depends on two inputs, the likelihood-based statistic and the number of previously observed claim periods $n$. Since $g(\cdot)$ is an unstructured function, we can use a non-linear formulation for its estimation, such as basis decomposition using B-splines.  In the case that one wants to impose monotonicity on the function, one can consider shape-constrained B-splines as in \cite{pya2015shape}. 

Note that we can write such specific functional form as:
$$
\hat{\Pi}_i^p \approx \textrm{Normal}(\Pi_i^p, se^2_i)
$$
$$
\log( \Pi_i^p) = \log(\Pi_i) + g( \mathcal{L}_i(\mathbf{Y}_n,\tilde{\theta}_i), n_i ),
$$
which can be identified as a generalized regression model with Gaussian response and a log-link function, where the simulated Bayesian premiums $\hat{\Pi}_i^p$ are the response variable, $\log(\Pi_i)$ is an offset, and both $\mathcal{L}_i(\mathbf{Y}_n,\tilde{\theta}_i)$ and $n_i$ are features with a joint non-linear effect.  

To estimate the parameters of this model, it is standard to minimize the mean square error (MSE) using any statistical learning techniques such as Additive Models, Neural Networks, Tree-based methods, etc. The MSE is given by $MSE= \frac{1}{M}\sum_{i=1}^M \left( \hat{\Pi}^p_i - \Pi_i\exp( g( \mathcal{L}_i(\tilde \theta_i), n_i ) ) \right)^2$. Gradient descent methods are commonly used to solve this optimization problem, which is not complex as only two features are involved.  Note that the least squares process is not heavily reliant on the normality of the premium estimator. Nonetheless, the latter does serve as motivation for this approach in contrast to other methods. 

The values of $\tilde \theta_i$ are also determined by the algorithm. As discussed in Section \ref{theindex}, tuning these values allows the likelihood-based statistic to better capture the information of the claim history, leading to better approximation capability of the surrogate model.  Therefore, we can set them as those values that provide the best approximation to the surrogate model. A challenge is that the values of $\tilde \theta_i$  are policyholder dependent and therefore it is necessary to have a systematic approach to link a policyholder to a reasonable value $\tilde \theta_i$. As attributes of the policyholder are also policyholder-specific, one may use them as a proxy to generate an appropriate $\tilde \theta$ according to a flexible assignation rule. Therefore, we may view the $\tilde \theta$ as a function of the set of parameters of the model, which include the policyholder attributes. Mathematically, a certain structure of the form $\tilde \theta_i = h(\mathcal{O}_i)$, for some unknown function $h(\cdot)$, can be assumed.   Do note that this construction does not imply that the latent variable of the Bayesian model, i.e. $\Theta$, depends on the policyholder's covariates, only the tuning parameters will have such dependence. 

Along those lines, any flexible technique can be used to find a data-driven function $h(\cdot)$, especially as the $\tilde \theta$ values are not intended to be interpreted on their own. The estimation of $\tilde \theta_i$ must be achieved jointly with the function $g(\cdot)$ and incorporated as part of the optimization process when fitting the model for $g(\cdot)$. This can be achieved directly in a tailored algorithm specified by the user or in an iterative fashion. It may be easier to consider the latter approach to use some of the already implemented methods in software packages. The algorithm iteration stops when the interpolation error, as measured by the MSE, can no longer be decreased. Along these lines, the general scheme for an iterative estimation is shown in algorithm \ref{algo}.

\begin{algorithm}[!htbp]
\caption{Fitting of $g( \cdot) $ and $\theta_i$}\label{alg:cap}
\begin{algorithmic}
\State $MSE \gets \textrm{Tol} +1 $
\State $\tilde \theta_i  \gets \textrm{Random number}  ~ \forall i=1, \ldots, M  $ \Comment{Start with random values for  $\tilde \theta_i$}
\While{$ MSE \ge \textrm{Tol}$} 
     \State $\mathcal{L}_i(\tilde \theta_i) \gets \sum_{j=1}^n \log f(Y_{i,j} \vert \Theta = \tilde{\theta}_i)  ~ \forall i=1, \ldots, M  $ \Comment{Compute the likelihood-based statistic }
    \State $ g( \cdot ) \gets \argmin_{g} \sum_{i=1}^M \left( \hat{\Pi}^p_i - \Pi_i\exp( g( \mathcal{L}_i(\tilde \theta_i), n_i ) ) \right)^2 $ \Comment{ Fit the function $g( \cdot )$ via LS }
    \State $\tilde \theta_i^o \gets \argmin_{\tilde \theta_i} \left( \hat{\Pi}^p_i - \Pi_i\exp( g( \mathcal{L}_i(\tilde \theta_i), n_i ) ) \right)^2 ~ \forall i=1, \ldots, M
    $ \Comment{Find  pseudo observations $\tilde \theta_i$}
    \State $h(\cdot) \gets \argmin_{h} \sum_{i=1}^M \left( \tilde \theta_i^o - h(\mathcal{O}_i) \right)^2 $ \Comment{ Fit the model $h( \cdot )$ via LS}
    \State $\tilde \theta_i \gets h(\mathcal{O}_i) ~~ \forall i=1, \ldots, M  $ \Comment{Upgrade the values $\tilde \theta_i$ using the fitted values of $h( \cdot )$ }
    \State $MSE \gets  \sum_{i=1}^M \left( \hat{\Pi}^p_i - \Pi_i\exp( g( \mathcal{L}_i(\tilde \theta), n_i ) ) \right)^2  $ \Comment{Upgrade current interpolation error}
\EndWhile
\end{algorithmic}
\label{algo}
\end{algorithm}

\subsection{Assessment of the surrogate function}
\label{assesment}

Before using the fitted formula for any analysis, it is important to evaluate the accuracy of the out-of-sample predictive power. In this section, we provide an overview of this task but refer the interested reader to \cite{sobester2008engineering} for further details.

In general, the quality of the interpolation can be evaluated by the coefficient of determination $R^2=1-\frac{MSE}{MST}$, where $MST = \frac{1}{M}\sum_{i=1}^M \left( \hat{\Pi}^p_i - \bar{\hat{\Pi}}^p \right)^2$. This measure seems to be a standard approach to assessing surrogate models. The coefficient of determination is always in the interval $[0,1]$, with 0 indicating poor interpolation and 1 indicating a perfect one. Therefore the fit can be easily interpreted in terms of percentages. A rule of thumb is to have a surrogate model with $R^2 \ge 0.90$.

To further assess the interpolation accuracy of the fitted formula, one can perform standard goodness of fit checks on the sub-portfolio of policyholders, akin to those performed in any linear regression model. This involves residual analysis (e.g., scatter plots, histograms with errors concentrated around zero, lack of bias, etc) and evaluation of error metrics (e.g., high coefficient of determination, small relative errors, small MSE, etc.). If the formula does not perform satisfactorily under this assessment, the functional form of $\hat{ \tilde G}_{\Pi}(\cdot)$ must be revised and changed to a more flexible structure to improve the accuracy.

Similarly, to evaluate the out-of-sample predictive power, one can use ideas from cross-validation methods in predictive modeling. For instance, the representative sub-portfolio of policyholders can be split into train and test samples. The surrogate model is fitted only on the former set and used to predict on the latter set. If the goodness of fit metrics on the train set are similar to those on the test set, then the out-of-sample predictive power of the formula is verified. However, if these two differ significantly, in the sense that the errors in the test set are considerably larger than in the train set, then the surrogate model is overfitting and therefore not reliable for extrapolation. In this case, the sample size of representative policyholders must be increased in order to fit such a flexible model. This can be achieved iteratively by gradually increasing the sample size, say 1\% at a time, until both the in-sample and out-of-sample performance of the surrogate model is verified.

\subsection{Extrapolation}
\label{extrapolation}

After the surrogate model is evaluated, computing the Bayesian premiums for the remaining policyholders in the portfolio is a simple matter of evaluating the surrogate function for each policyholder. This approach significantly simplifies the computation workload for large portfolios, reducing it to only a minimal portion of the portfolio.

Certain surrogate models provide a way to construct confidence intervals for the prediction.  For instance, if a Gaussian Process is used, confidence bands can be constructed based on the posterior predictive distribution. However, this approach is model-dependent. There is some research in the machine learning literature for constructing confidence bands that can be adapted to surrogate modeling. For example, works such as \cite{kumar2012bootstrap} and \cite{kabir2018neural} aim to construct confidence bands in a model-agnostic fashion. Nevertheless, we note that this task is still under research and refer the reader to those papers if there is a need to construct them.

\section{Simulation study}
\label{simulation}
In this section, we illustrate via a simulation study how our newly defined likelihood-based statistic can capture the information on the claim history of a policyholder. We do so by investigating the achievable accuracy of the surrogate model at approximating the premiums constructed on different choices of the model-prior distributions and premium principles commonly observed in actuarial modeling. We separate the analysis depending on whether the variable of interest is continuous or discrete. Results are summarized in Tables \ref{tablediscrete} and \ref{tablecontinous}, in which we display the coefficient of determination $R^2$ obtained from the fitted surrogate model.

The simulation study's general setup involves generating synthetic portfolios of policyholders and their claim history from the Bayesian models listed in Tables \ref{tablediscrete} and \ref{tablecontinous}. To emulate real insurance portfolios, we create a heterogeneous portfolio of 50,000 policyholders. The mean of the model distribution for each policyholder  $\mu$  is linked to a synthetic systematic effect $\alpha$ (emulating effect of policyholder individual attributes ) and a single latent variable $\Theta$ affecting the mean of the distribution as follows:
$$
\log(\mu) = \alpha +\Theta
$$

We assume non-randomness of the dispersion parameter in the model distribution in the cases where there is such a parameter.  We note that this structure is consistent with traditional Bayesian regression models used in insurance as commented in Section \ref{framework}. The values of $\alpha$ and the dispersion parameters are chosen to resemble those that are usually obtained when fitting a model. Therefore, these mimic the typical heterogeneity present in real datasets and the usual frequency and severity patterns observed in insurance.

To generate the claim history of each policyholder, we utilize the resulting value $\mu$ obtained from the previous step along with the chosen model distribution. In order to maintain simplicity, we consider a scenario where each policyholder has a claim history spanning $n=5$ periods.

We compute the Bayesian premiums via the importance sampling algorithm discussed in Section \ref{importance}, utilizing 20,000 samples for each policyholder so that the error of estimation through simulation is negligible. As our aim in this section is to evaluate the predictive power of the surrogate model, we fit $\hat{\tilde G}(\cdot)$ using an unstructured function that is estimated through multidimensional B-splines, which can be achieved via a generalized additive model (GAM) implementation.
 
To specify the selection of the values $\tilde \theta$ which vary from one policyholder to another, we utilize a \emph{random forest} structure for the function $h(\cdot)$ that establishes a connection between the values of $\tilde \theta$ and both the parameters of the model distribution and the claim history of policyholders. Specifically, we employ the systematic component $\alpha$ values of each policyholder as the input features for the random forest, such that $\tilde \theta = h(\alpha)$. 
We then proceed with the implementation of the estimation via Algorithm \ref{algo}. Notably, the computational cost per iteration proved to be nearly insignificant, thanks to the efficient fitting of the Generalized Additive Model and the Random Forests. This efficiency was largely due to the relatively modest size of the set of representative policyholders and the low dimensionality of the features in the surrogate model as accomplished via the likelihood-based statistic.

During our simulations, we found that the algorithm rapidly achieved a desirable level of interpolation within just a few iterations, well before converging to the results detailed below. It is worth mentioning that once the values of $\tilde \theta$ approached a range relatively close to convergence, the surrogate model consistently produced similar values for the coefficient $R^2$ iteration after iteration. This suggests that the convergence of the function $g(\cdot)$ appeared to occur more swiftly than that of the function $h(\cdot)$. Consequently, as $\tilde \theta$ values neared a region sufficiently close to the optimum, their impact on the outcome diminished. This behavior can be attributed to the fact that once we are in proximity to the optimal selection of $\tilde \theta$ values, the additional information gained from the summary statistic becomes less substantial in terms of improving the model fit. Hence, it is worthwhile to consider the possibility of early termination of the algorithm if the achieved results at a certain point are already highly satisfactory.

We then proceed to evaluate the precision of the surrogate function by comparing the premiums estimated by the model with the Bayesian premiums, using the coefficient of determination $R^2$. The results are presented in Tables \ref{tablediscrete} and \ref{tablecontinous}.

\begin{table}[!htbp]
\centering
\begin{tabular}{ccccc}
\hline \hline
\multicolumn{1}{l}{\textit{\textbf{}}} & \multicolumn{1}{l}{\textbf{}} & \multicolumn{3}{c}{Premium Principle} \\ \cline{3-5}
Model &  & Expected Value & Standard deviation & Exponential \\ \hline 
Poisson-Gamma &  & 99.63\% & 99.63\% & 99.63\% \\
Poisson-LogNormal &  & 99.89\% & 99.89\% & 99.89\% \\
NegBinom-InvGaussian &  & 99.59\% & 99.64\% & 99.99\% \\
Logarithmic-Weibull &  & 99.84\% & 99.85\% & 99.84\% \\
GammaCount-Weibull &  & 97.07\% & 96.98\% & 97.08\% \\
Gen-Poisson-Lognormal &  & 99.95\% & 99.96\% & 99.95\% \\ \hline \hline
\end{tabular}
\caption{Coefficient of determination $R^2$  of the surrogate model on the entire portfolio for simulation study on discrete distributions}
\label{tablediscrete}
\end{table}

\begin{table}[!htbp]
\centering
\begin{tabular}{cccc}
\hline \hline
\multicolumn{1}{l}{\textit{\textbf{}}} & \multicolumn{1}{l}{\textbf{}} & \multicolumn{2}{c}{Premium Principle} \\ \cline{3-4}
Model &  & Expected Value & Standard deviation  \\ \hline 
Gamma-Gamma &  & 96.11 \% & 96.12\% \\
Lognormal-InvGaussian &  & 94.52\% & 94.52\% \\
LogLogistic-Lognormal &  & 99.83\% & 99.83\% \\
InvGaussian-Weibull &  & 99.88\% & 99.90\% \\
Pareto-Gamma &  & 99.92\% & 99.92\% \\
Burr-Lognormal &  & 99.67\% & 99.67\% \\
\hline \hline
\end{tabular}
\caption{Coefficient of determination $R^2$  of the surrogate model  on the entire portfolio for simulation study on continuous distributions}
\label{tablecontinous}
\end{table}

Our results demonstrate that the $R^2$ values are close to 99\% in almost all scenarios, indicating that the fitted premiums closely resemble the Bayesian premiums, regardless of the selected model distributions and premium principles. We observed in further simulations that the performance of the fitted formula fluctuates more when the portfolio exhibits different levels of heterogeneity with respect to the policyholder attributes and claim history. However, the surrogate model was still able to produce similar results as the ones above even in such scenarios. Some of the distributions in the study are not part of the exponential dispersion family, yet however, the surrogate model showed a desirable performance at reproducing such premiums. Therefore we conclude that the likelihood-based statistic can effectively summarize almost all of the relevant information regarding a policyholder's claim history, enabling the creation of an appropriate surrogate model.

\section{Numerical illustration with real data }\label{realdata}

In this section, we demonstrate the use of the surrogate model on a real data set obtained from a European automobile insurance company. The data set consists of policyholders' contract information spanning from January 2007 to December 2017, with claim frequencies from two business lines: Third Party Liability insurance (TPL) and Physical Damage (PD). The data set contains detailed information on the policyholder or their automobiles, such as car weight, engine displacement, engine power, fuel type (gasoline or diesel), car age, and age of the policyholder.

The contract of interest for this application is a policy with two coverages: TPL and PD. The number of claims in both lines may be dependent, as the same car accident can lead to claims in both lines of business. It should be noted that not all policyholders have a fully observed history in the two lines. Policyholders may initially start with a contract on one policy and then upgrade to obtain a policy that covers the two of them, or vice versa (see Figure \ref{descriptive}). Therefore, the claim history for some policyholders can be considered partially observed in the sense that the number of claims for a particular line may not be available for certain periods.

Table \ref{summary} and Figure \ref{descriptive} present key summary statistics on the dataset analyzed in this study. The number of claims in both lines of business exhibits a behavior consistent with established patterns in insurance, namely, over-dispersion, a high frequency of zero claims, and a significant correlation between the two lines of business. The large portfolio is composed of 184,848 policyholders, among which only 41,956 have fully observed exposure in both lines of business simultaneously.

\begin{table}[!ht]
\begin{tabular}{cccccc}
%\cline{3-6}
\hline
\hline
 &  & \multicolumn{4}{c}{\textbf{Claim Frequency}} \\ 
 \cline{3-6}
\multicolumn{1}{c}{\textbf{TYPE}} & \textbf{Number of Policyholders} & \multicolumn{1}{c}{\textbf{Mean}} & \multicolumn{1}{c}{\textbf{Median}} & \multicolumn{1}{c}{\textbf{Variance}} & \textbf{Correlation} \\ \hline
\multicolumn{1}{c}{\textbf{TPL}} & 138,923 & \multicolumn{1}{c}{0.038} & \multicolumn{1}{c}{0.000} & \multicolumn{1}{c}{0.044} & \multicolumn{1}{c}{0.245} \\ 
\multicolumn{1}{c}{\textbf{PD}} & 87,517 & \multicolumn{1}{c}{0.306} & \multicolumn{1}{c}{0.000} & \multicolumn{1}{c}{0.512} &  \\ 
\multicolumn{1}{c}{\textbf{TOTAL}} & 184,484 & \multicolumn{1}{c}{0.162} & \multicolumn{1}{c}{0.000} & \multicolumn{1}{c}{0.279} &  \\ \hline
\hline
\end{tabular}
\caption{Summary statistics of the dataset}
\label{summary}
\end{table}

As depicted in Figure \ref{descriptive}, a significant proportion of policyholders renew their contract with the company, enabling the observation of long claim histories, with some policyholders having up to seven years of exposure. Therefore, it is essential to account for such observed experiences in the ratemaking process.

\begin{figure}[!h]
     \centering
     \begin{subfigure}[b]{0.45\textwidth}
         \centering
         \includegraphics[width=\textwidth]{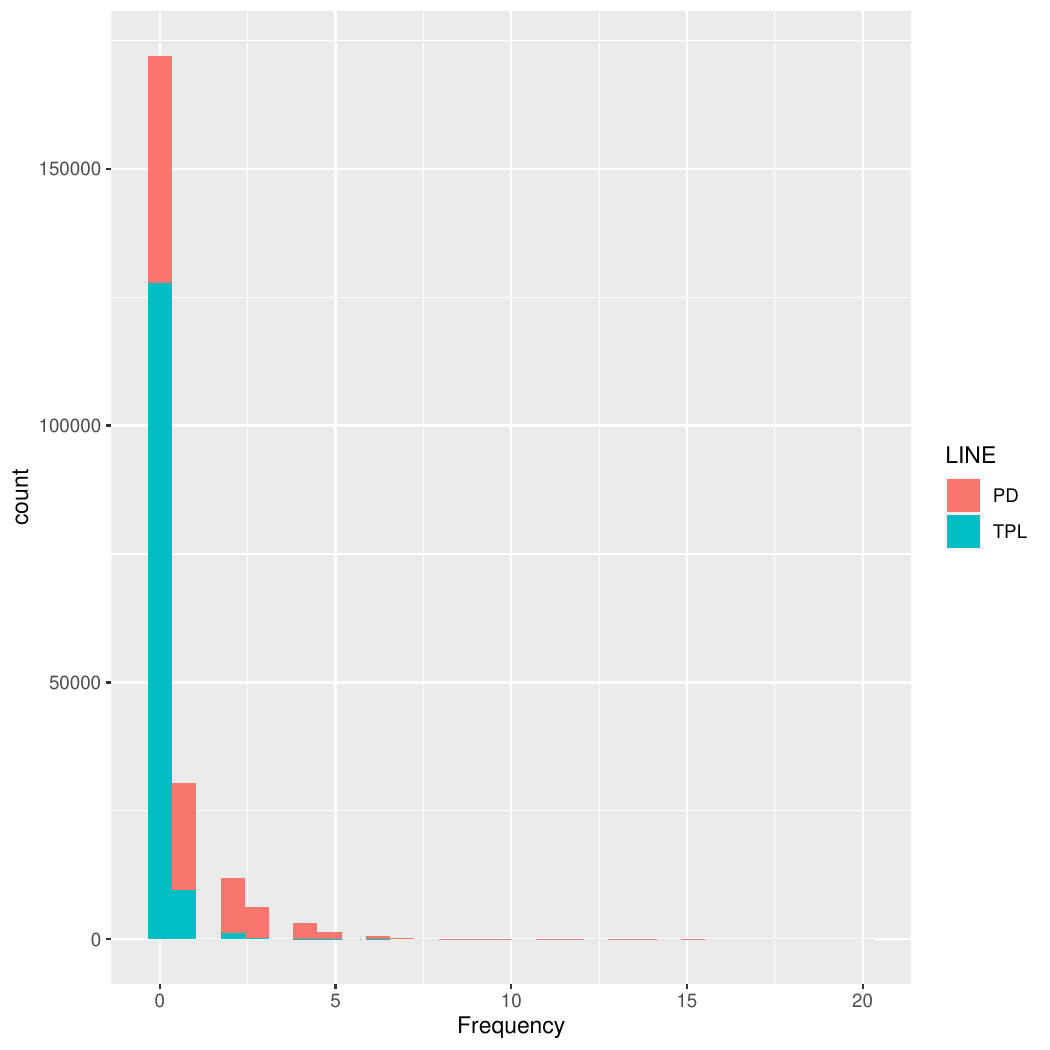}
         %\caption{Distribution number of claims per year}
    \end{subfigure}
     \hfill
     \begin{subfigure}[b]{0.45\textwidth}
         \centering
         \includegraphics[width=\textwidth]{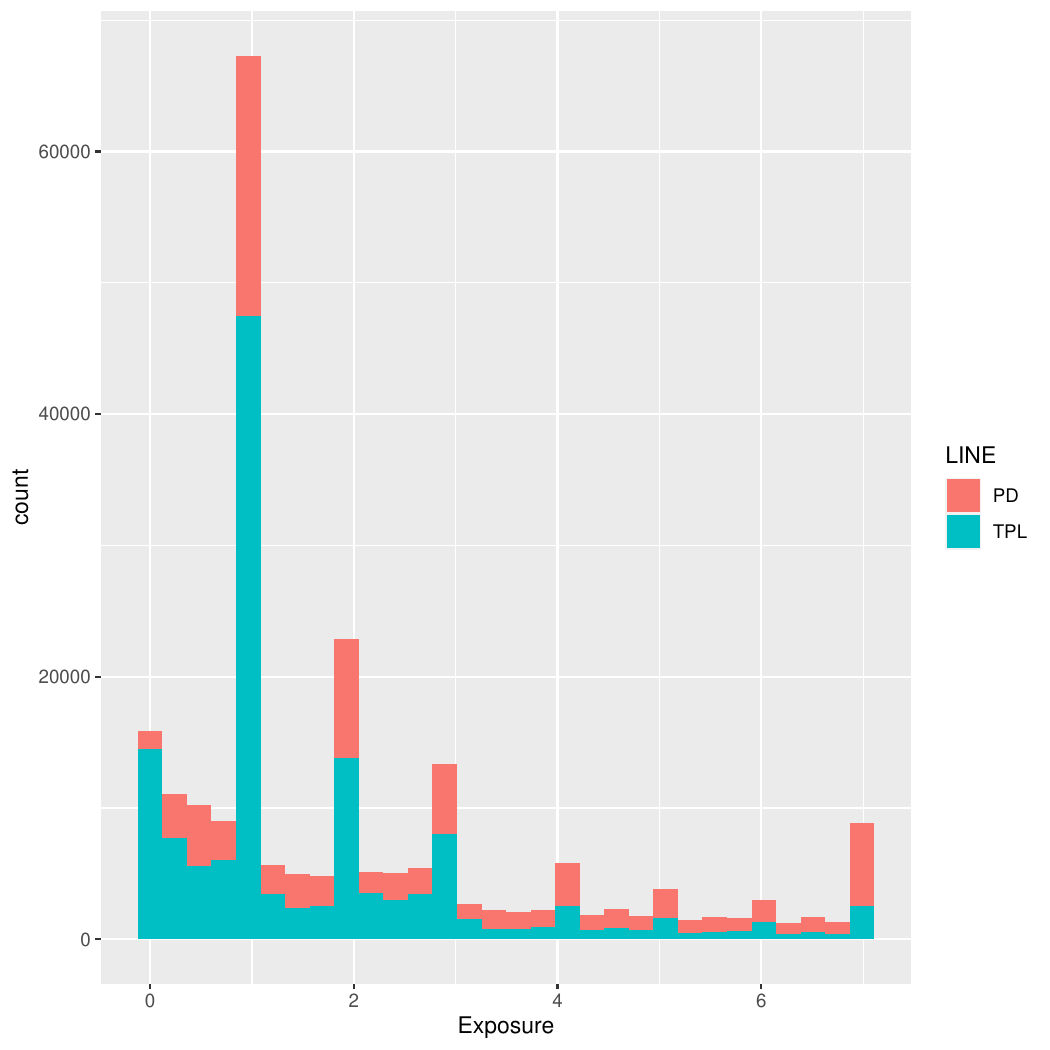}
         %\caption{Distribution of the exposure of a policyholder (in years)}
         %\label{exposure}
     \end{subfigure}
            \caption{Distribution of claim frequency (left) and exposure (right) per claim type}
       \label{descriptive}
\end{figure}

Given the aforementioned observations, the insurance company intends to conduct an experience rating analysis of the claims experience of policyholders at the time of contract renewal. Specifically, we consider the application of the \emph{Exponential} premium principle as an example, with a 5\% surcharge, which is the method employed by the insurance company to establish actuarial premiums.

\subsection{The Bayesian model and quantities of interest}

We illustrate the use of the surrogate model with a bivariate mixed negative-binomial regression model proposed by \cite{tzougas2021multivariate}, which is very flexible and can capture over-dispersion and dependent frequencies.

Let $Y_j^{(d)}$ be the number of claims from a given policyholder in year $j$, associated with the $d$th line of business, with $d=1$ being PD and $d=2$ being TPL. Let $x$ be the information of covariates, and let $\beta^{(d)}$ be the associated vector of regression coefficients for the $d$-th coverage. Similarly, let's denote with $\omega^{(d)}$  the time exposure of the contract for each coverage. We consider the hierarchical model
$$
Y_j = \left( \begin{matrix}
    Y_j^{(1)}\\
    Y_j^{(2)}
\end{matrix} \right)   \sim_{iid} f(y \vert \Theta, \langle \boldsymbol{x}, \boldsymbol{\beta} \rangle) = \textrm{NegBinom}(y^{(1)} ; \mu^{(1)}  \Theta , r^{(1)} )*\textrm{NegBinom}(y^{(2)} ; \mu^{(2)}  \Theta , r^{(2)} )
$$
where, for $d=1,2$
\small
$$
\log \mu^{(d)} = \log \omega^{(d)}+ \beta_0^{(d)}+\beta_1^{(d)} \textrm{CarWeight}+\beta_2^{(d)}\textrm{EngineDisplace}+\beta_3^{(d)}\textrm{CarAge}+\beta_4^{(d)}\textrm{Age}+\beta_5^{(d)}\textrm{EnginePower}+\beta_6^{(d)}\textrm{Fuel}
$$
\normalsize

\noindent and $\Theta \sim P(\theta)  = \textrm{InvGauss}(1, \sigma^2)$. We use the notation $\textrm{NegBinom}(y; \mu, r)$ to denote the probability mass function of a negative binomial with mean $\mu$ and dispersion $r$. Similarly, $\textrm{InvGauss}(1, \sigma^2)$ denotes an Inverse-Gaussian distribution with mean 1 and variance $\sigma^2$.

It is worth noting that the proposed model employs a shared latent variable for each line of business. This simplifies the task of working with partially observed data since it enables the estimation of $\Theta$ for both business lines even when a policyholder has observations in only one of them.

We emphasize that the model lacks analytical expressions for both the posterior and the predictive distribution. As such, numerical methods are necessary to obtain any desired quantity of interest, as discussed in detail in \cite{tzougas2021multivariate}. The estimation of the model parameters is conducted in R using generalized linear mixed models with a Negative-Binomial response. The resulting fitted parameters are presented below:

\begin{table}[!htbp]
\footnotesize
\begin{tabular}{cccccccccc}
\hline
\hline 
\textbf{Variable}         & \textbf{$\beta_0$} & \textbf{$\beta_1$} & \textbf{$\beta_2$} & \textbf{$\beta_3$} & \textbf{$\beta_4$} & \textbf{$\beta_5$} & \textbf{$\beta_6$} & \textbf{$r$} & \textbf{$\sigma^2$} \\ \hline
$Y^{(1)}$    &  -2.78 & $-4.48*10^{-5}$ & $2.69*10^{-5}$ & $-2.94*10^{-2}$ & $-4.20*10^{-3}$ & $2.24*10^{-3}$ & $-8.23*10^{-2}$ & 0.86 & 0.58 \\ 
$Y^{(2)}$  & -1.397 & $14.19*10^{-5}$ & $6.10*10^{-5}$ & $-0.20*10^{-2}$            & $-7.94*10^{-3}$  & $6.06*10^{-3}$ & $-22.08*10^{-2}$ &  0.86 & 0.58  \\ \hline
\hline
\end{tabular}
         \caption{Estimated Parameters of the Bivariate mixed  NegBinomial regression model }
\end{table}

\subsection{ Estimation of the surrogate model}

Here we proceed to calculate the Bayesian premium using the exponential principle
$$
\Pi(Y_{n+1} \vert \mathbf{Y_n})  = \frac{1}{0.05} \log( E \left( \exp(0.05*(Y_{n+1}^{(1)}+Y_{n+1}^{(2)})) \vert \mathbf{Y_n} \right)).
$$ 

To handle the large size of the portfolio, we employ a surrogate model. Specifically, we utilize the \emph{cube method} (implemented via the \texttt{samplecube()} function in \texttt{R}) to extract a sub-portfolio comprising approximately 5\% of the overall portfolio, which amounts to roughly 9,224 policyholders. To ensure a balanced representation, we match the sub-portfolio with the overall portfolio in terms of the average number of claims for PD and TPL, which capture the claim history, as well as the average fitted values of $\mu^{(1)}$ and $\mu^{(2)}$, which reflect the policyholder's attributes. The computational cost of this process is minimal as can be appreciated in Table \ref{times}.

We use the importance sampling method described in Section \ref{importance} to estimate premiums, with 50,000 samples to ensure the accuracy of these values. Note that this step is only necessary for the sub-portfolio of $5\%$ policyholders, however, we perform the computationally expensive simulation for the entire portfolio to enable a comparison of the premiums obtained from the surrogate model. We would like to highlight in terms of the simulation time that a single replication on the entire portfolio takes approximately 18 times longer than a replication on the representative policyholder, as illustrated in Table \ref{times} below. This is close to the empirical ratio of 20 associated with the proportion of 5\% vs 100\% of the representative portfolio. The simulation took place in a large server with 32 cores, 125GB RAM memory and CPU 2 x Intel E5-2683 v4 Broadwell @ 2.1GHz. 

\begin{table}[H]

\centering
        \begin{tabular}{ccccc}
        \hline
        \hline 
        \textbf{Process} & \textbf{Total Portfolio} & \textbf{Representative Policyholders}   \\ \hline
        Selecting sample & - & 32.94 \\
        Simulation for the premium & 919,008.00 & 50,976.00 \\
        Fitting surrogate function & -  & 2,480.22 \\
Extrapolation & -  & 3.39 \\ \hline
Total Time  & 919,008.00 ($\approx$ 255 h) & 53,492.55 ($\approx$ 15 h)         \\ \hline
        \hline
      
        \end{tabular}

 \caption{Comparison of the CPU Time (in seconds) required for the calculation of premiums }
 \label{times}
 
\end{table}

As previously discussed, the likelihood-based statistic can be subdivided into sub-statistics that represent summary statistics, with one for each business line in the model. In this case, we can identify two likelihood-based sub-statistics, denoted by $d=1,2$, and defined as follows:
$$
\mathcal{L}^{(d)} (\mathbf{Y}_n; \tilde \theta ) = \sum_{j=1}^n  \log \left( \textrm{NegBinom} \left( Y_j^{(d)}; \mu^{(d)} \exp( \tilde \theta ), r^{(d)} \right)\right).
$$

Regarding the surrogate function, we emphasize the rating factor functional as described in Equation (\ref{eqn4}) to achieve a transparent ratemaking process. Nonetheless, we also explore the use of an unstructured surrogate function via multidimensional B-Splines for comparison purposes. Since the model under consideration is multivariate, with $Y_j=(Y_j^{(1)}, \ldots, Y_j^{(D)})$, we can further decompose the function $g(\cdot)$ in Equation (\ref{eqn4}) to measure the effect of the likelihood-based sub-statistics separately. Thus, it is possible to quantify the importance of the claim history of each component on the resulting Bayesian premium.

For instance, if we opt for an additive structure for the function $g(\mathcal{L}(\mathbf{Y}_n; \tilde{\theta}),n)$, our surrogate function can be expressed as:
$$
\hat{\tilde G}_\Pi(\mathcal{L} (\mathbf{Y}_n; \tilde{\theta} ),n, \mathcal{O}) = \Pi(Y_{n+1})  \exp( g_1(\mathcal{L}^{(1)} (\mathbf{Y}_n; \tilde{\theta} ))+g_2(\mathcal{L}^{(2)} (\mathbf{Y}_n; \tilde{\theta} ))+g_3(n^{(1)}) +g_4(n^{(2)}) ) +c . 
$$
where $ \Pi(Y_{n+1}) $ is the manual premium or current premium, which in this case is given by the expression  $\Pi(Y_{n+1})  = \frac{1}{0.05} \log( E \left( \exp(0.05*(Y_{n+1}^{(1)}+Y_{n+1}^{(2)}))  \right))$, and the $g_j(\cdot)$ are non-linear functions that depend only on the argument given. 

In this setup, the function $g_1(\cdot)$  provides a measure of the effects of the past experience in the number of PD claims in the Bayesian premium, $g_2(\cdot)$ provides a measure of the effects of the TPL number of claims in the Bayesian premium, and $g_3(\cdot)$ and $g_4(\cdot)$ provide the effect of the number of periods for which the policyholder has been observed for each line of business. It should be noted that some policyholders may have only partially observed information, meaning that their claim history may differ for each line of business. Therefore, separate effects for $n^{(1)}$ and $n^{(2)}$ are considered.   Lastly, we added the intercept $c$ to the surrogate function, which can be thought of as an overall macro-adjustment to premiums in the portfolio.

To estimate the functions $g_k(\cdot)$ and the tuning parameters $\tilde \theta$ for each individual, only $5\%$ of the policyholders are used, as described previously. We use a B-Slines representation for the functions $g_j( \cdot)$,  and the fit is accomplished via a generalized additive model (GAM) with Gaussian response and with the additive structure above.  The values of $\tilde \theta$ are estimated using a random forest with $\mu^{(1)}$ and $\mu^{(2)}$, which depend on the policyholder attributes, as features.  Once again, recall we use the policyholder attributes as a proxy to define policyholder-specific values systematically. 

Such estimation of the functions $g_k(\cdot)$ and the values $\tilde \theta$ is performed iteratively, as described in Section \ref{estimation}.  The behavior of the algorithm is analogous to the one described in the simulation study in  Section \ref{simulation}. It is worth mentioning that both the GAM and random forest models exhibit flexibility and require minimal computational resources for training, making them suitable for surrogate modelling. In fact, the entire process of fitting the surrogate model is completed in a minimal fraction of time when compared to the approach using simulation for the whole portfolio as can be observed in Table \ref{times}.  

Lastly, it is observed that the total computational time required to calculate the premiums using the two approaches differs significantly. The approach employing representative policyholders and the surrogate function takes 17 times less time than the direct simulation approach. Furthermore, it is important to highlight the substantial difference in computational resources required for both methods. The direct simulation approach demands higher-than-average computational power, including more cores and RAM memory, while the surrogate model approach can be effectively executed on a personal computer.

\subsection{Results and interpretation}
The distribution of the estimated values $\tilde{\theta}_i$ is depicted on the left-hand side of Figure \ref{thetas2}. It is evident that there exists significant variability in these parameters, indicating a bimodal pattern attributed to the heterogeneity of the portfolio. This observation is reinforced by the figure on the right, illustrating the distributions of both manual premiums and Bayesian premiums, which also display a comparable bimodal trend.

 \begin{figure}[!htbp]
     \centering
     \begin{subfigure}[b]{0.45\textwidth}
         \centering
         \includegraphics[width=\textwidth]{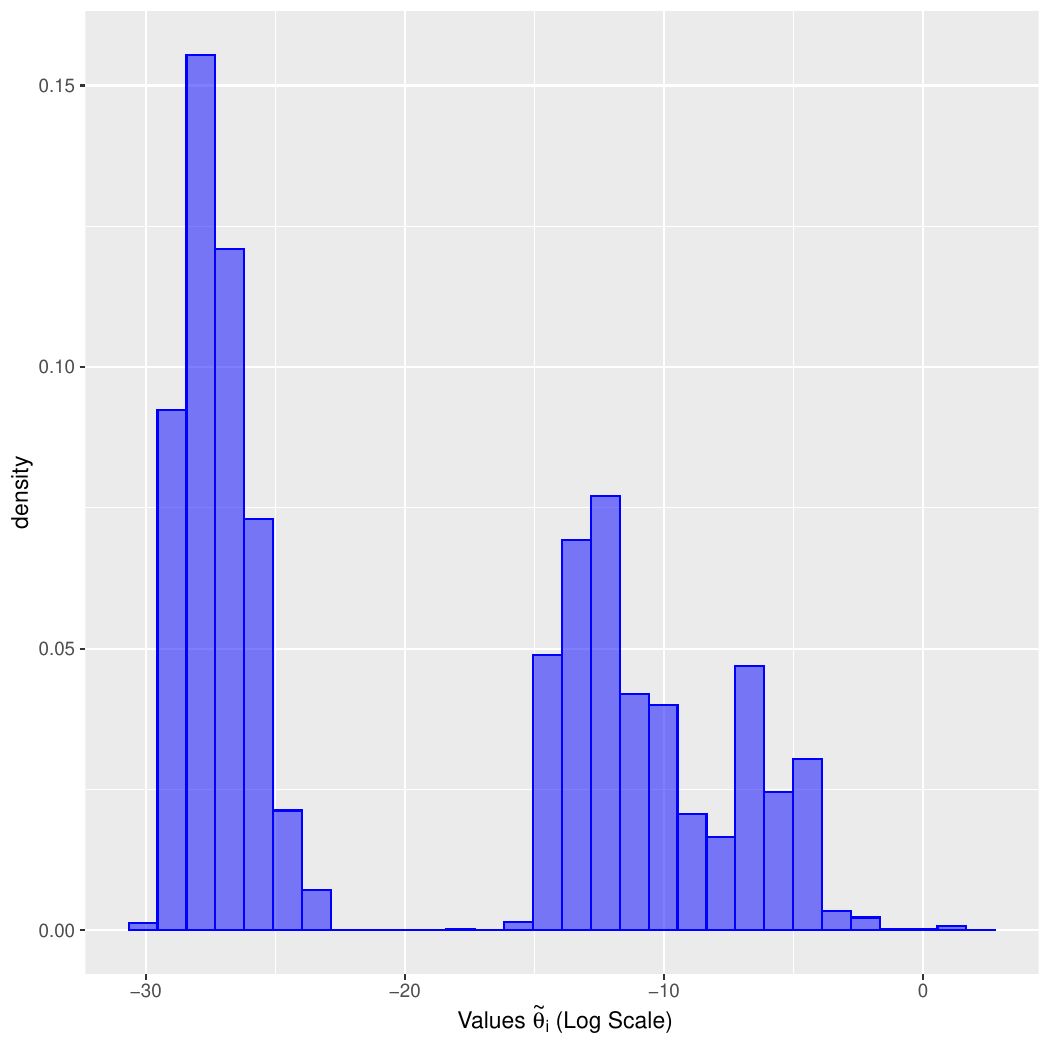}
         %\caption{Distribution of $\tilde \theta_i$}
          
     \end{subfigure}
     \hfill
     \begin{subfigure}[b]{0.45\textwidth}
         \centering
        \includegraphics[width=\textwidth]{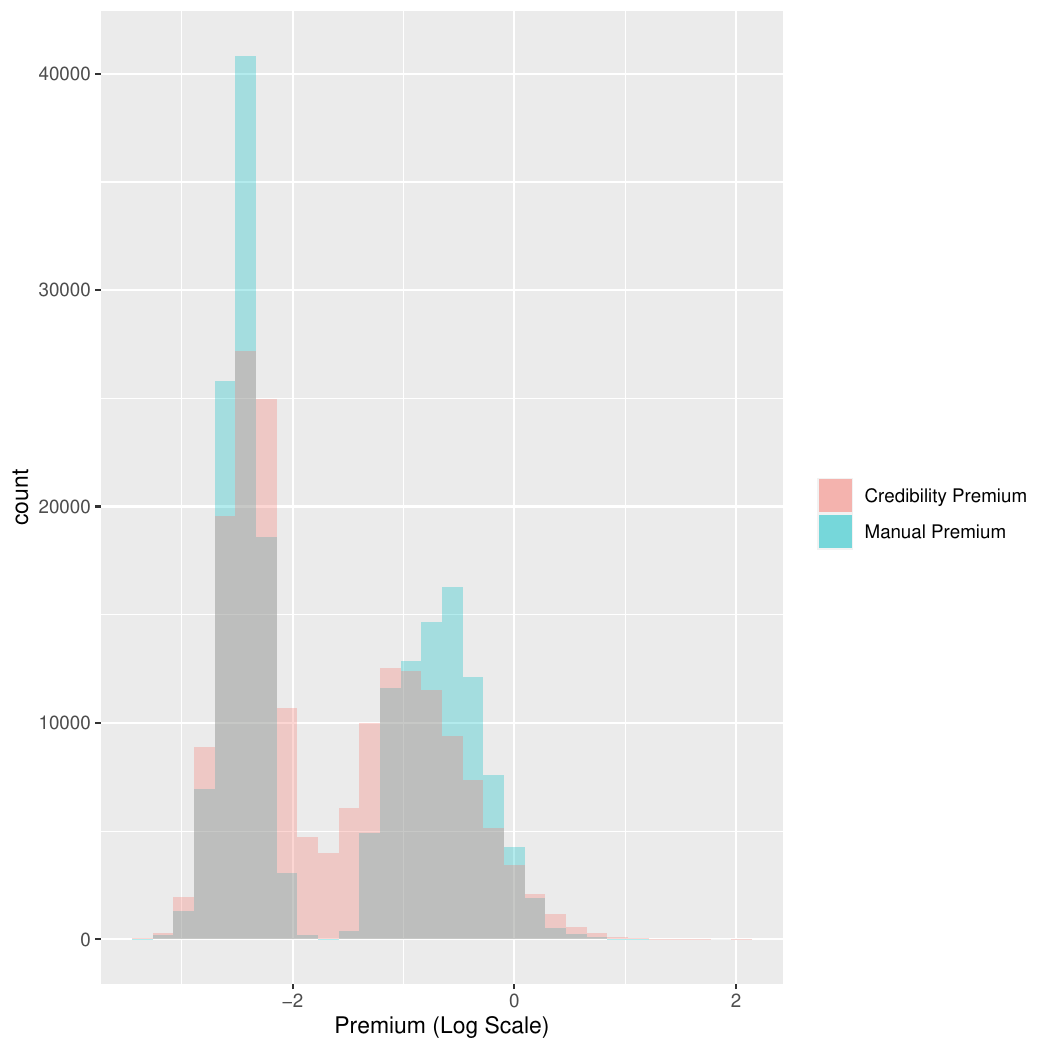}
         %\caption{$\tilde \theta_i$ vs  manual premiums (log scale)}
          
     \end{subfigure} 
     
        \caption{  Histogram of the estimated values of $\tilde \theta_i$ (Left) and Histogram of the premiums (Right) }
        \label{thetas2}
\end{figure}

Figure \ref{distindex} displays the distribution of the likelihood-based sub-statistics among policyholders, which exhibits a negative exponential-like behavior. Notably, the peak of the distribution is observed around 0 and coincides with those policyholders having no claims, which are the vast majority. In contrast, policyholders who have made claims fall within the body of the distribution, and those with the highest number of claims usually fall in the left tail.

It is important to acknowledge that the relationship between the value of the likelihood-based statistic and the number of claims is not one-to-one due to the influence of policyholder attributes on the definition of the summary statistic itself. However, we can interpret values of the likelihood-based statistic close to 0 as generally associated with ``low-risk" policyholders, while significantly negative values are indicative of ``risky" policyholders. Therefore, a distribution of the summary statistic concentrated around 0 suggests a low-risk portfolio, whereas a distribution with elongated left tails indicates a risky portfolio that requires careful attention.

 \begin{figure}[!htbp]
     \centering
     \begin{subfigure}[b]{0.45\textwidth}
         \centering
         \includegraphics[width=\textwidth]{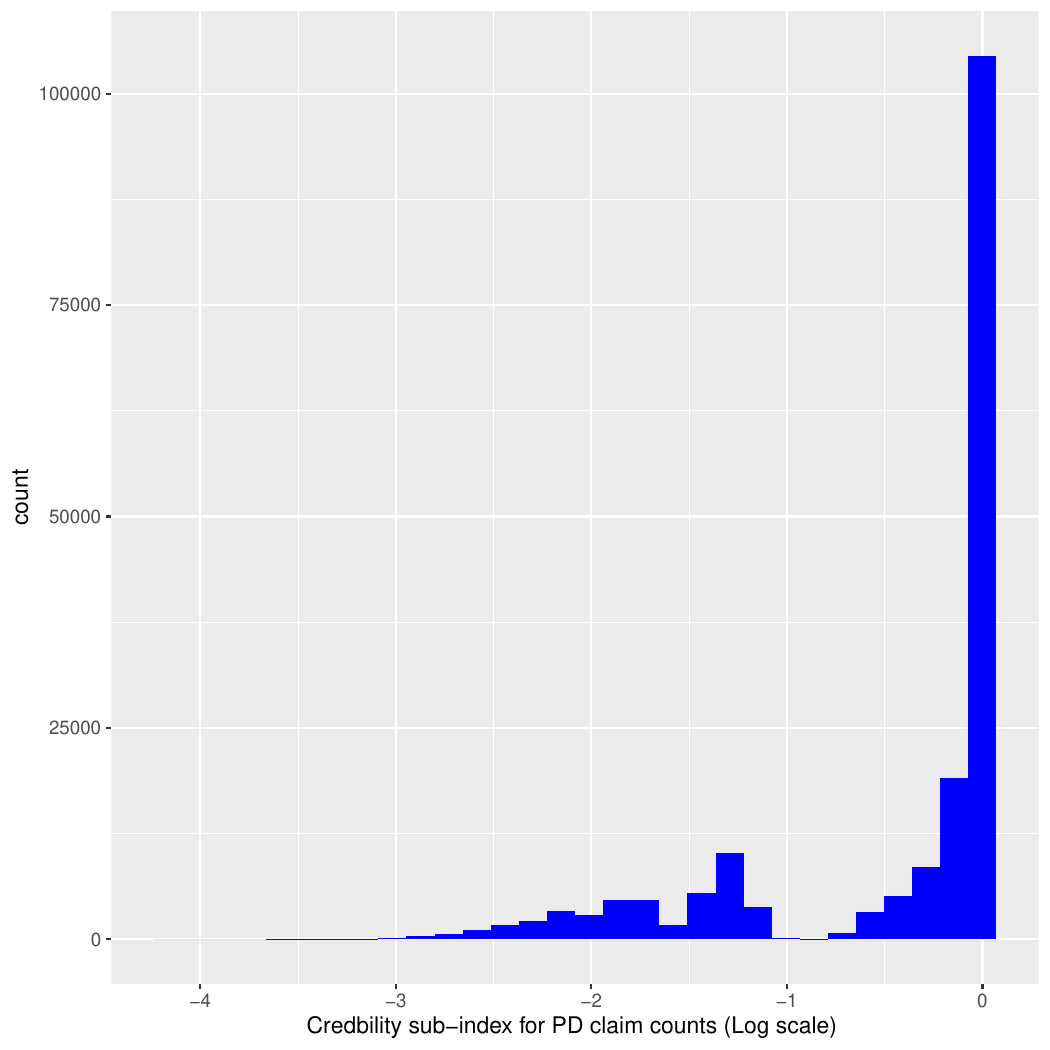}
         %\caption{Distribution for PD claims}
          
     \end{subfigure}
     \hfill
     \begin{subfigure}[b]{0.45\textwidth}
         \centering
        \includegraphics[width=\textwidth]{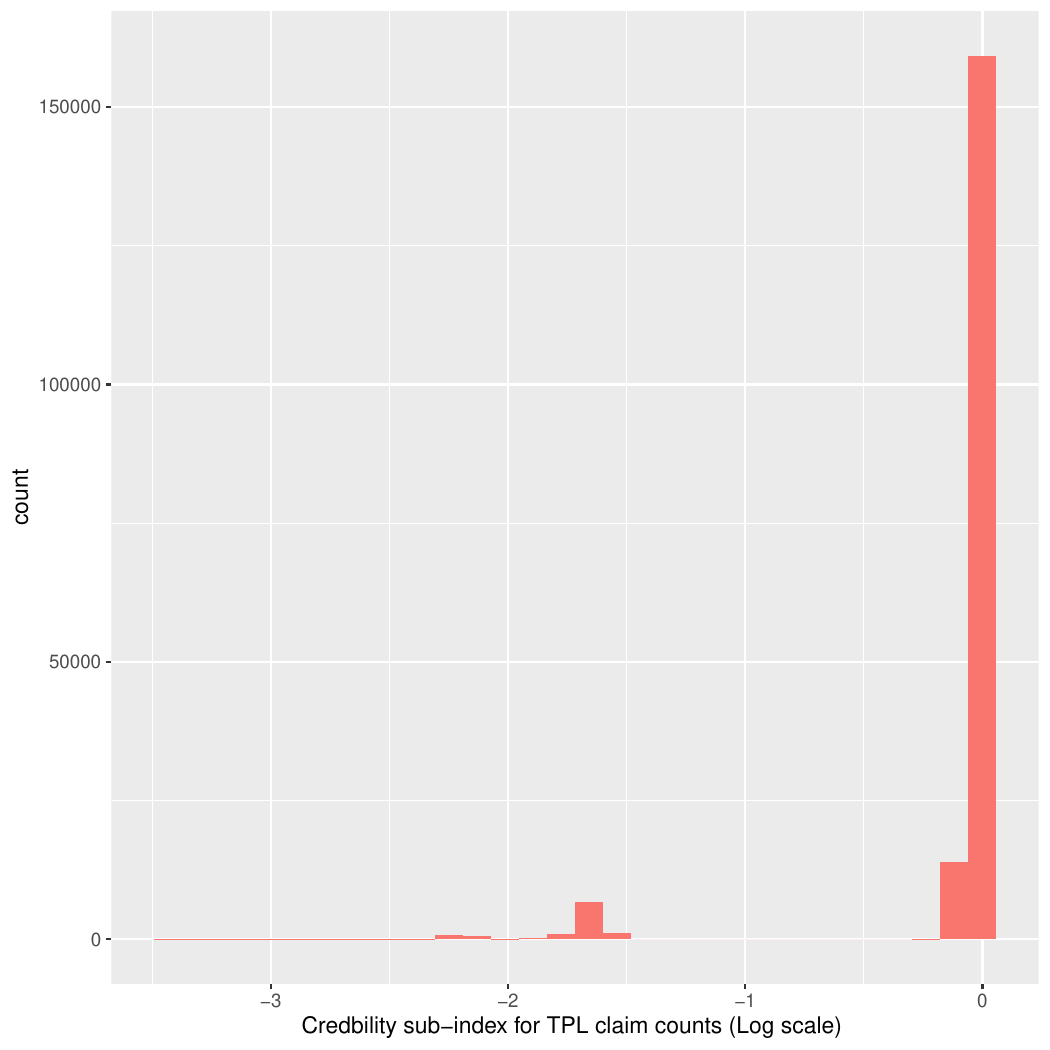}
       %\caption{Distribution for TPL claims}
          
     \end{subfigure} 
        \caption{Distribution of the likelihood-based sub-statistics. Left is for PD number of claims and Right for TPL number of claims}
        \label{distindex}
\end{figure}

We now proceed to show the fitted surrogate function.  We obtained a fitted value $c=-0.0093$, which indicates an overall reduction in the premiums on the portfolio, although almost insignificant in magnitude. Moreover, the estimated functions $g_k( \cdot )$ are presented in Figure \ref{Estimatedparameters2}. The results indicate that the greater the likelihood-based sub-statistics, the less the estimated Bayesian premium for a specific policyholder. This interpretation is consistent with the one we obtained in Figure \ref{distindex}.

We have that the function $\exp(g_1(\cdot))$ exhibits a null-effect at a value of $\mathcal{L}^{(1)} ( \mathbf{Y}_n; \tilde{\theta} )$ around -13, whereas for the function $\exp(g_2(\cdot))$, the null-effect is observed at around -15 for a value of $\mathcal{L}^{(2)} ( \mathbf{Y}_n;\tilde{\theta})$. Thus, policyholders with likelihood-based sub-statistics less than -13 (for $\mathcal{L}^{(1)} ( \mathbf{Y}_n; \tilde{\theta})$) and less than -15 (for $\mathcal{L}^{(2)} ( \mathbf{Y}_n;\tilde{\theta}$) tend to have a higher than expected number of PD claims and TPL claims, respectively. Consequently, they require a revised premium that is larger than the manual premium.  Similarly, policyholders with values of likelihood-based sub-statistics greater than -13 (for $\mathcal{L}^{(1)} ( \mathbf{Y}_n;\tilde{\theta})$) and greater than -15 (for $\mathcal{L}^{(2)} (\mathbf{Y}_n; \tilde{\theta})$) tend to have a lower than expected number of PD claims and TPL claims, respectively, and may receive a revised premium that is lower than the manual premium. 

The range of the effect is larger for the function $\exp(g_2(\cdot))$ (i.e., from 0.3 to 1.6) than for the function $\exp(g_1(\cdot))$ (i.e., from 0.5 to 1.4), indicating that the experience on the number of TPL  claims may produce a larger change on the Bayesian premiums than the experience in the number of PD claims. This result is intuitive as TPL claims have a significantly lower frequency than PD claims (see Table \ref{summary}), and thus, having a TPL claim would have a more substantial impact on the Bayesian premium than a PD claim.  Nevertheless, the effect on the premium depends on the specific combination of claim experience of the policyholder.  

The two plots at the bottom of Figure \ref{Estimatedparameters2} depict the effect of the number of periods of exposure on the resulting Bayesian premium. While the impact of these quantities is not directly interpretable for the insurance company,   these are required as part of the predictive model to distinguish short and long claim histories. 

 \begin{figure}[!htbp]
     \begin{subfigure}[b]{0.45\textwidth}
         \centering
         \includegraphics[width=\textwidth]{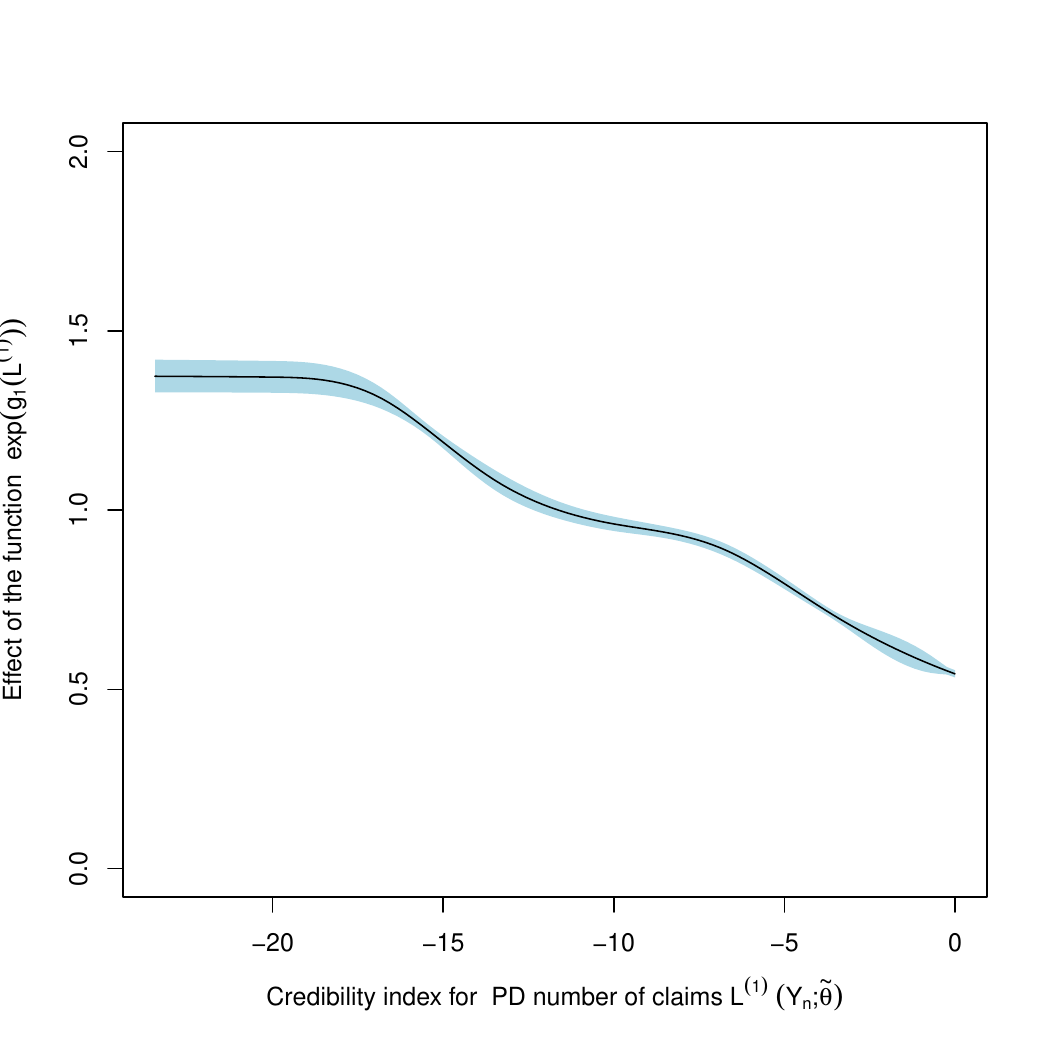}
         %\caption{Estimated function $g_1$}
          
     \end{subfigure}
     \hfill
     \begin{subfigure}[b]{0.45\textwidth}
         \centering
         \includegraphics[width=\textwidth]{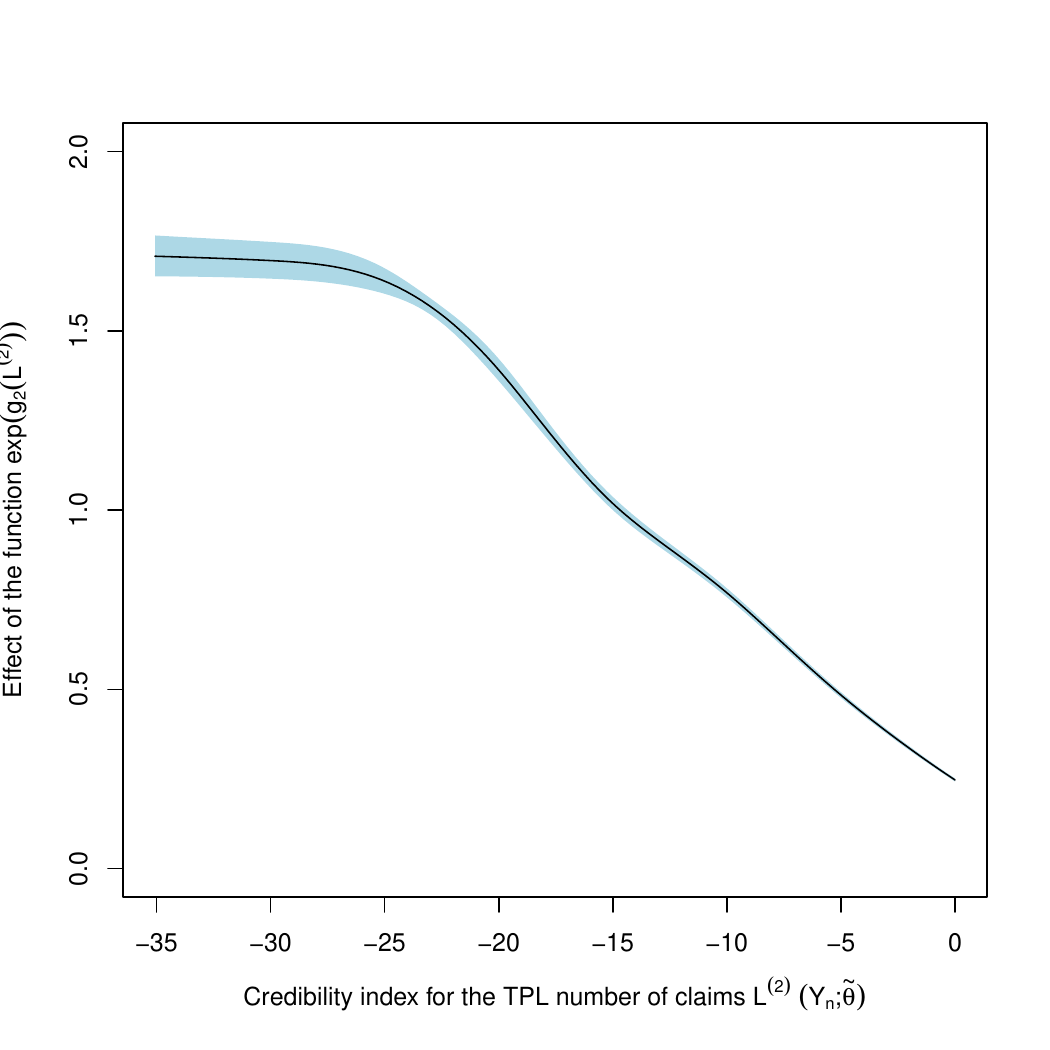}
         %\caption{Estimated function $g_2$}
          
     \end{subfigure}
     \hfill 
%     \end{figure}
 %\begin{figure}[!htbp]
 %\ContinuedFloat
     \begin{subfigure}[b]{0.45\textwidth}
         \centering
         \includegraphics[width=\textwidth]{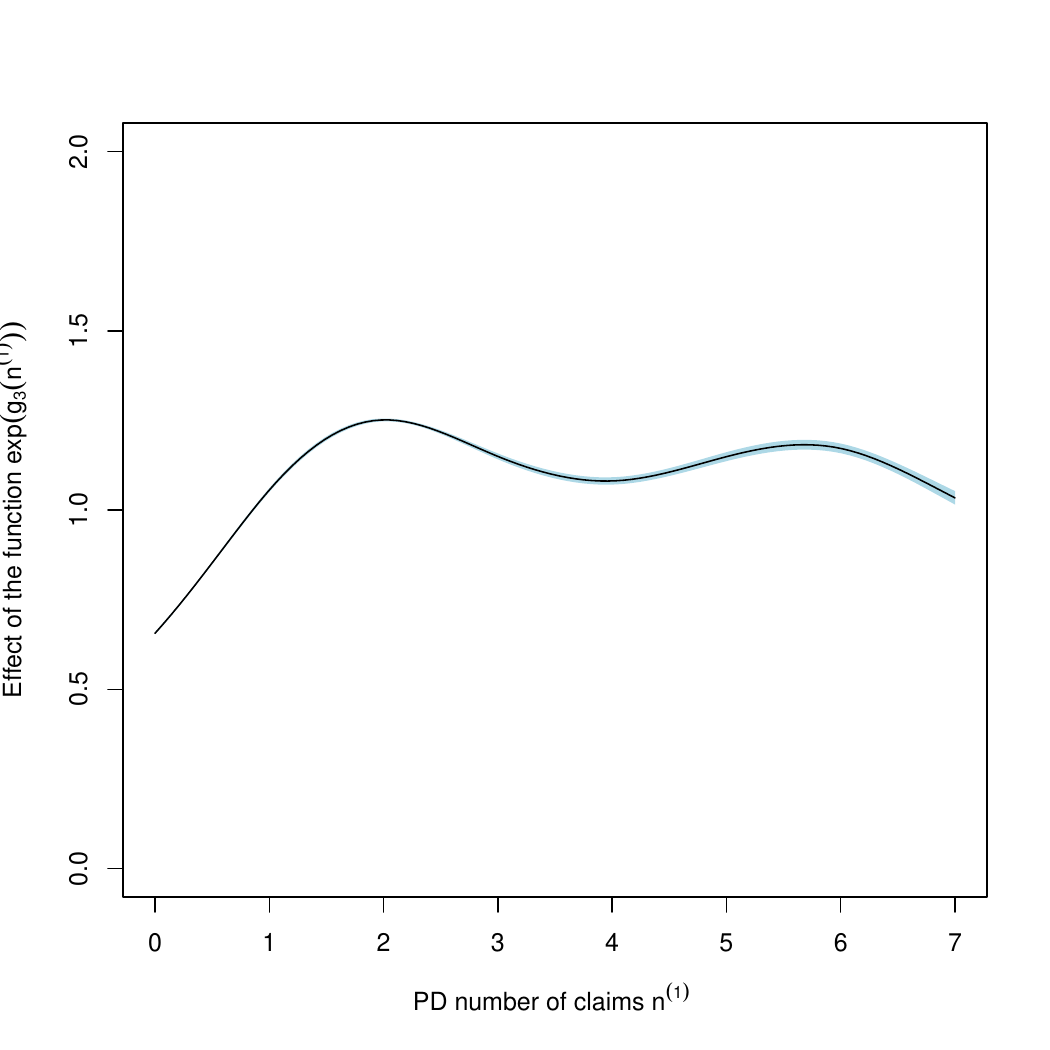}
         %\caption{Estimated function $g_3$}
          
     \end{subfigure}
     \hfill
     \begin{subfigure}[b]{0.45\textwidth}
         \centering
         \includegraphics[width=\textwidth]{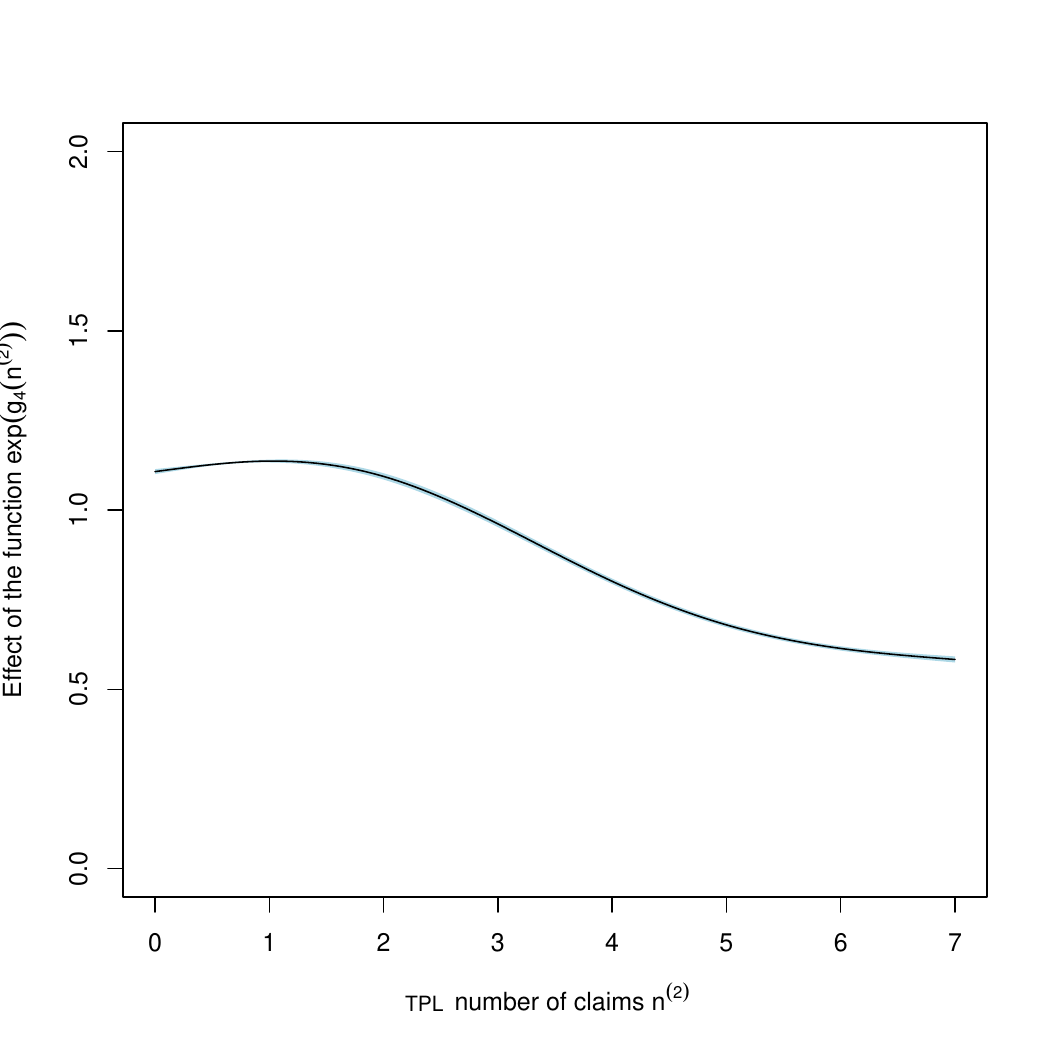}
         %\caption{Estimated function $g_4$}
          
     \end{subfigure}
%     \hfill
        \caption{Estimation of the functions $g(\cdot)$. Top-left is $g_1(\cdot)$, Top-right is $g_2(\cdot)$, Bottom-left is $g_3(\cdot)$ and Bottom-right is $g_4(\cdot)$ }
        \label{Estimatedparameters2}
\end{figure}

 \subsection{Goodness of fit}
 We now proceed to evaluate the suitability of the surrogate function as presented in Section \ref{assesment}. The evaluation involves two main tasks: verify the accuracy of the surrogate function at the Bayesian premiums and assess its reliability for extrapolation in the out sample when only $5\%$ of the portfolio is utilized. 
 
To assess the accuracy of the surrogate function, we compare the fitted values of the surrogate model with the Bayesian premiums. Figure \ref{gof3} on the left-hand side presents a jitter plot of the two premiums. The darker the color, the more concentration of points in that region. The plot indicates that our pricing formula produces premiums that are a close approximation to the Bayesian premiums, with the points closely clustered around the 45-degree line. Note that while some points lie outside the reference line, the vast majority concentrate close to it indicating a good performance of the surrogate.  To further evaluate the accuracy of the surrogate function, Table \ref{gof2} presents error metrics of the fitted premium versus the Bayesian premiums evaluated at both the policyholder level and the aggregate level. The interpolation results in an overall coefficient of determination of $R^2=0.97$, indicating that the pricing formula reproduces around 97\%  of the premiums variance, which is an overall descent level of interpolation. The fitted premiums deviate from the Bayesian premiums on average by only 3\%, and the mean error of interpolation is almost zero, implying no bias in the resulting estimation.

\begin{figure}[!htbp]
     \centering
     \begin{subfigure}[b]{0.45\textwidth}
         \centering
         \includegraphics[width=\textwidth]{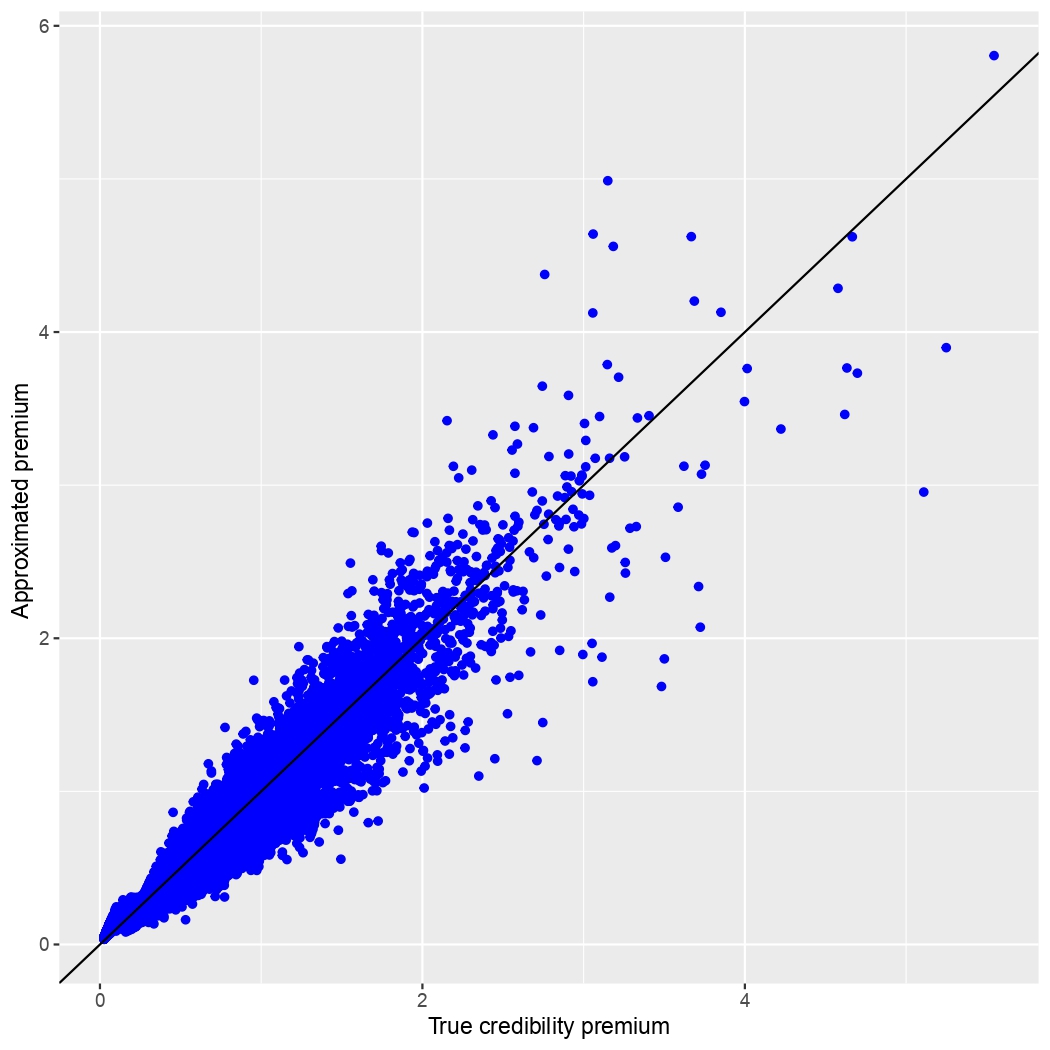}
         %\caption{True Premiums vs Fitted Premium}
          
     \end{subfigure}
     \hfill
     \begin{subfigure}[b]{0.45\textwidth}
         \centering
        \includegraphics[width=\textwidth]{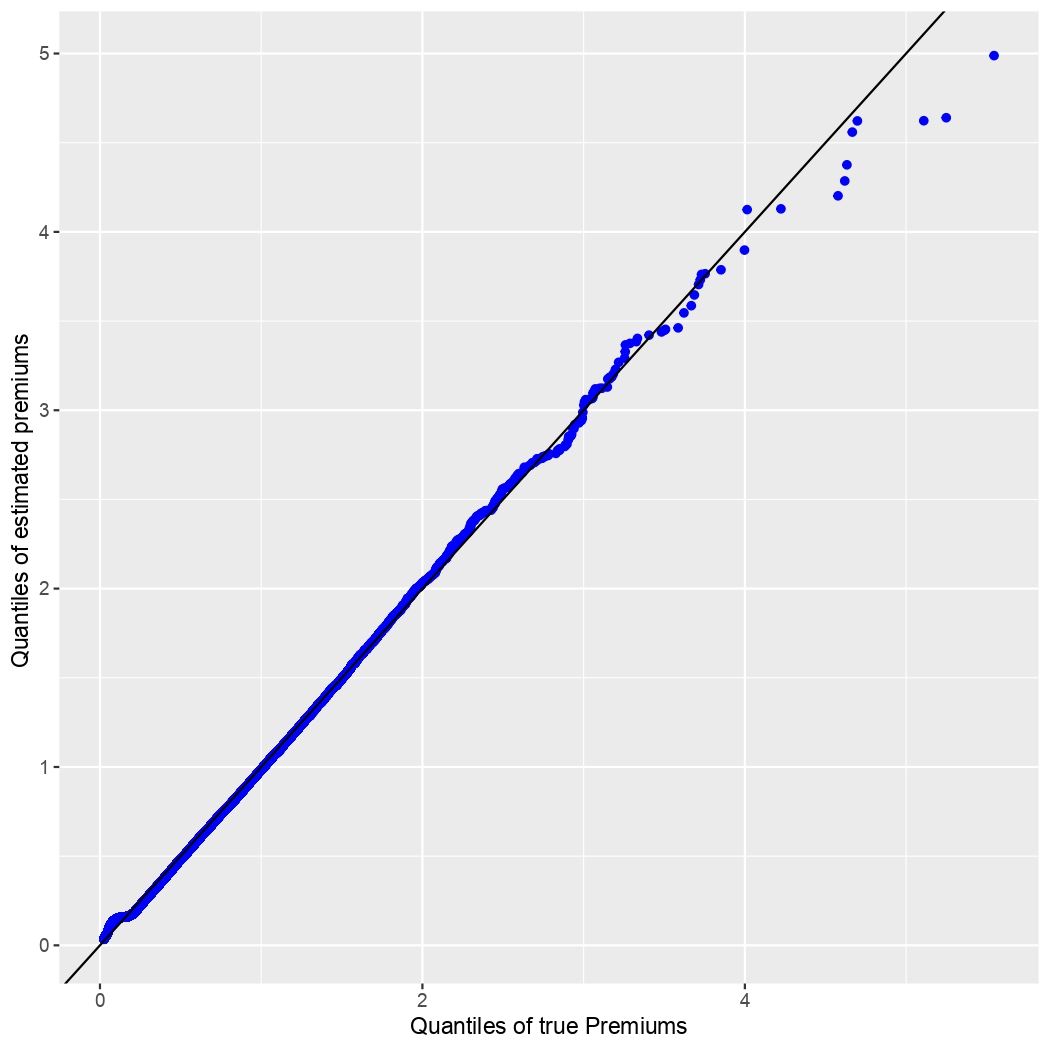}
       %  \caption{QQ-plot estimated vs true premiums }
          
     \end{subfigure} 
     \caption{  Comparison of approximated premiums vs true premiums. Left-hand side dispersion plot and Right-hand side QQ plot.  }
     \label{gof3}
     \end{figure}

To assess the reliability of extrapolation, we need to evaluate the out-of-sample performance of the predictive model on a test data set and ensure that it exhibits behavior consistent with the in-sample data. Specifically, in this application, as we computed the premiums for the entire portfolio, we used the remaining 95\% as our test data.   Table \ref{gof2} displays the error metrics for both the in-sample and out-of-sample policyholders. In this case, we observe that the error metrics for the out-of-sample policyholders are comparable to those obtained for the in-sample policyholders, with no appreciable differences. Hence, extrapolating the predictive power of the model from only 5\% of the portfolio is sufficient to provide reliable predictions for the entire portfolio using the surrogate model.

\begin{table}[!htbp]
          \centering
         
         \begin{tabular}{llllllll}
         \hline
        \hline 
        & \multicolumn{3}{c}{At Individual level} & & \multicolumn{3}{c}{At Aggregate level} \\ \cline{2-4}  \cline{6-8}         
        Sub-portfolio & ME  & MAE & MAPE & & $R^2$ &  Error & MPE  \\ \hline
        Out of sample & -0.0004 & 0.0233  & 0.093 & & 0.972 & 59.41 & 0.0035 \\
        In sample & -0.0006 & 0.0234 &  0.093 & & 0.973 & 46.57 & 0.0027       \\ \hline
        \hline
        \end{tabular}
\caption{   Error metrics of the surrogate model with rating factor functional form. Individual level means the metric is calculated at the policyholder level and then averaged. At the portfolio level, the metric is calculated on the aggregate premiums directly. Abbreviations meaning are ME: Mean Error, MAE: Mean Absolute Error, MAPE: Mean Absolute Percentage Error, $R^2$: Coefficient of determination. Error: True - Predicted, MPE: Mean Percentage Error}
\label{gof2}
\end{table}

Similarly, from an insurance company's perspective, it is relevant to assess the premiums at an aggregate level to ensure the necessary solvency of the company. This means verifying that the distribution of the premiums resulting from the surrogate model is similar to that obtained from the Bayesian premiums. To evaluate this, we compare the two distributions using a QQ plot, which is presented on the right-hand side of Figure \ref{gof3}. The plot shows that the distribution of the fitted premiums is essentially the same as that of the Bayesian premiums.  Moreover, Table \ref{gof2}  shows a comparison of the error on the premiums but at the portfolio level (i.e. total real amount vs total fitted amount). The mean percentage error (MPE) shows that the discrepancy between the aggregation of premiums vs the one predicted by the surrogate model is insignificant with an error of less than 0.5\% for both the in-sample and out-of-sample. 

Therefore, we conclude that the pricing formula provides an accurate approximation of the Bayesian premiums, and the insurance company can rely on the fitted premiums without significant differences in further portfolio-level metrics, such as total premium earned across the portfolio, loss ratios, and other quantities of interest.

Finally, for the sake of comparison, we assess the results of the surrogate model using an unstructured functional form in Figure \ref{gof5} and Table \ref{gof_un}. This model can be considered the most flexible surrogate model since it is constructed without imposing any functional restrictions on the dependence of the likelihood-based statistics. As such, it serves as a baseline for evaluating predictive power. The surrogate model is fitted using multidimensional B-Splines representation for the function $\hat{\tilde G}_{\Pi}$, and using as inputs the manual premium, the two likelihood-based sub-statistics, and the number of known periods in each policy. The results are displayed in Figure \ref{gof5}  and Table \ref{gof_un}, in the same fashion as the ones illustrated for the surrogate model with rating functional form.

\begin{table}[!htbp]
\centering

        \begin{tabular}{llllllll}
         \hline
        \hline 
        & \multicolumn{3}{c}{At Individual level} & & \multicolumn{3}{c}{At Aggregate level} \\ \cline{2-4}  \cline{6-8}         
        Sub-portfolio & ME  & MAE & MAPE & & $R^2$ &  Error & MPE  \\ \hline
        Out of sample & -0.0005 & 0.0036  & 0.018 & & 0.99 & 71.78 & 0.0041  \\
        In sample & -0.0006 & 0.0036 &  0.017 & & 0.99 & 42.66 &   0.0024     \\ \hline
        \hline
        \end{tabular}
 \caption{  Error metrics of the surrogate model with unstructured surrogate function. Individual level means the metric is calculated at the policyholder level and then averaged. At the portfolio level, the metric is calculated on the aggregate premiums directly. Abbreviations meaning are ME: Mean Error, MAE: Mean Absolute Error, MAPE: Mean Absolute Percentage Error, $R^2$: Coefficient of determination. Error: True - Predicted, MPE: Mean Percentage Error}
 \label{gof_un}
\end{table}

Our results in Table \ref{gof_un} demonstrate that using an unstructured functional form provides a relatively better fit than the rating factor functional form, with lower error metrics and a higher coefficient of determination of 99\%. This implies an almost perfect interpolation, indicating that the fitted premiums are much closer to the true premiums than the previous surrogate. Similarly,  the graph on the left-hand side of Figure \ref{gof5} shows that the fitted premiums are much closer to the true premiums than the rating functional form, in the sense of much lower fluctuation around the 45-degree line. Analogously, on the right-hand side of Figure \ref{gof5}, the QQplot shows that the distribution of the fitted premiums still resembles the same pattern of the true premiums, and the results at the aggregate level in Table \ref{gof_un} show almost insignificant differences with the total premiums at the aggregate level.

\begin{figure}[!htbp]
     \centering
     \begin{subfigure}[b]{0.45\textwidth}
         \centering
         \includegraphics[width=\textwidth]{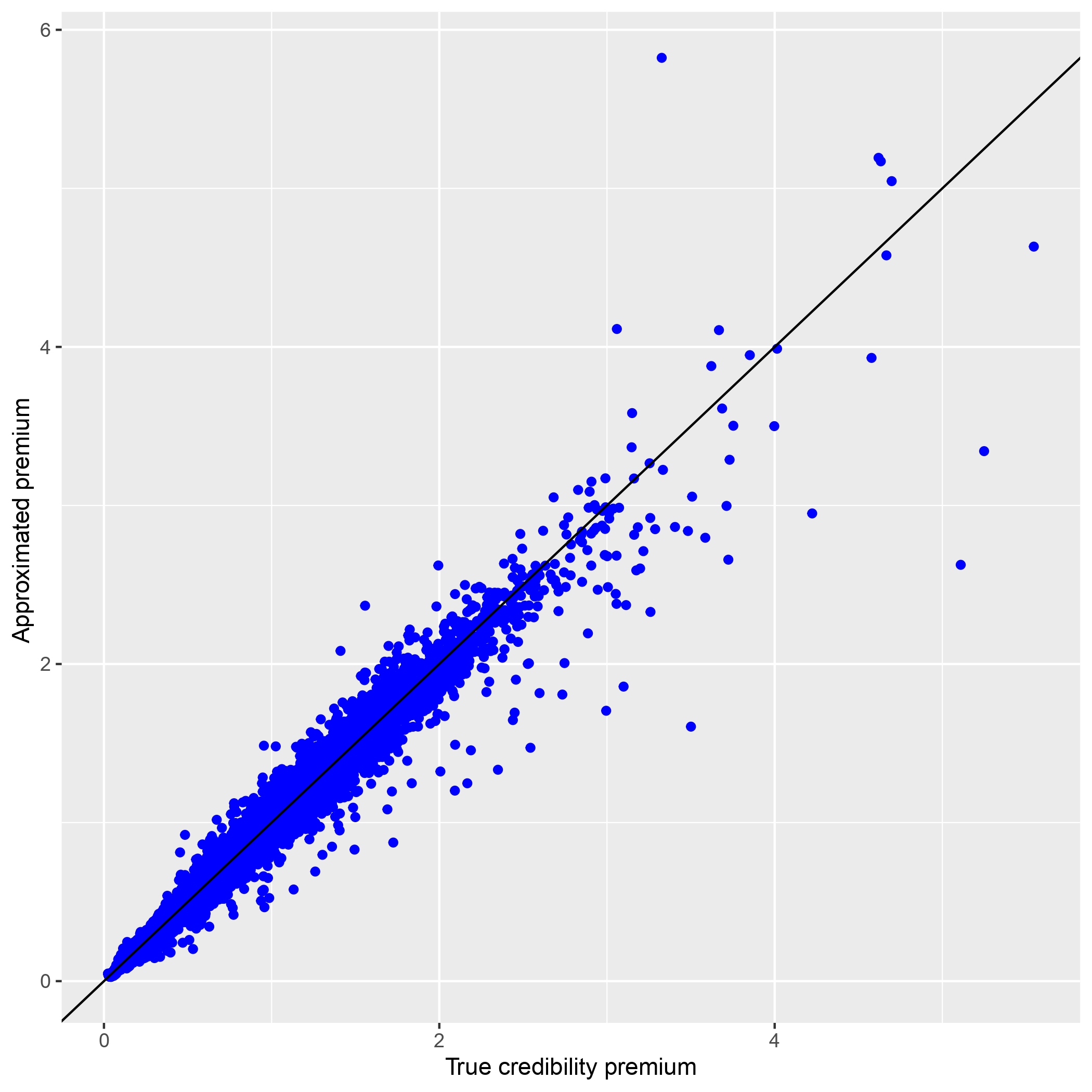}
         %\caption{True Premiums vs Fitted Premium}
          
     \end{subfigure}
     \hfill
     \begin{subfigure}[b]{0.45\textwidth}
         \centering
        \includegraphics[width=\textwidth]{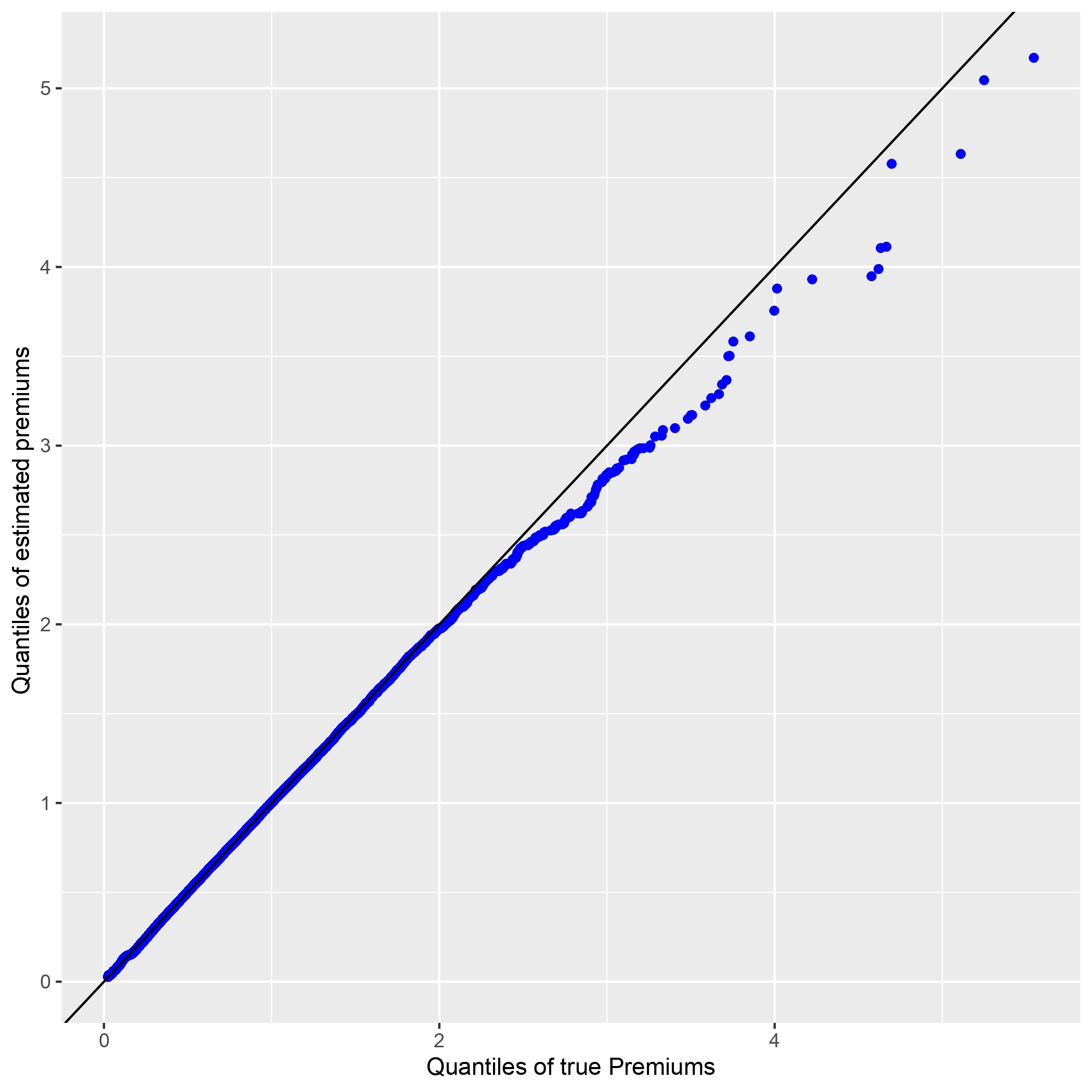}
       %  \caption{QQ-plot estimated vs true premiums }
          
     \end{subfigure} 
     \caption{  Comparison of approximated premiums vs true premiums for the unstructured model. Left-hand side dispersion plot and Right-hand side QQ plot.}
     \label{gof5}
     \end{figure}

We would like to note that, even though the surrogate model using the rating factor functional form provides a less competitive fit for the resulting premiums, it still provides  almost comparable results to the unstructured form, with the added benefit of providing a transparent ratemaking process.  Therefore, the choice between these two models will depend on the specific needs of the insurance company. Finally, we want to emphasize that the overall methodology via surrogate modeling inexpensively provides desirable results.

 \subsection{An example with two policyholders }

This subsection illustrates how the surrogate function is utilized for the experience rating process. We examine two policyholders who exhibit similar risk behavior, as indicated by their attributes, but possess different claim histories. These two policyholders have nearly identical predicted mean numbers of claims, as estimated by the model, and have held policies with the insurance company for a duration of $n=4$ years. The individual attributes of these policyholders are presented below.

\begin{table}[!htbp]
\small
\centering
\begin{tabular}{cccccccc}
\hline
\hline
\small
\textbf{Client ID}         & \textbf{CarWeight} & \textbf{EngDisp} & \textbf{CarAge} & \textbf{Age} & \textbf{EngPow} & \textbf{Fuel type} & \textbf{Expected Claims} \\ \hline
A    &  1515 & 1248 & 1 & 82 & 63& Gasoline & 0.245 (per year) \\ 
B  & 1475 & 1360 & 6 & 72          & 55  & Gasoline & 0.246 (per year) \\ \hline
\hline
\end{tabular}
\caption{Covariates of the Policyholders }
\end{table}

Policyholder A has not filed any claims in the past 4 years, neither in PD nor in TPL, while Policyholder B has filed two claims - one in PD and one in TPL. Consequently, it can be inferred that Policyholder A is less risky than Policyholder B. Intuitively, the premium for Policyholder A should decrease since no claims have been observed in 4 years, even though 1 claim is anticipated on average in this period (i.e., 4*0.25 = 1 claim). Conversely, the premium for Policyholder B should increase considerably as the total number of claims is twice the number expected from the model for this period (i.e., 2 claims vs. 4*0.25 = 1 claim).

Keeping this in mind, we further examine how the likelihood-based statistic and the fitted pricing formula reflect these two risk behaviors. Table \ref{tabexam} presents the likelihood-based statistic and the values $\tilde \theta_i$ for each policyholder.  While the values $\tilde \theta_i$ are in practice the same (because of the same predicted number of claims), note the difference in the magnitude of the likelihood-based sub-statistics. Specifically, for Policyholder A, who has not filed any claims, the sub-statistics exhibit relatively minimal values closer to 0, indicating a claim history that is less risky than expected, as per our interpretation of Figure \ref{Estimatedparameters2}. On the other hand, for Policyholder B, who has filed multiple claims, the sub-statistics exhibit the opposite behavior, with a large negative value indicating a claim history that is riskier than expected. Consequently, the likelihood-based statistic captures the distinct nature of the claim histories of these two policyholders.

\begin{table}[!htbp]
\centering
\begin{tabular}{cccc}
\hline
\hline
\small
\textbf{Client ID}         & {$\tilde \theta$} & $\mathcal{L}^{(1)}(\mathbf{Y}_n; \tilde \theta )$ & $\mathcal{L}^{(2)}(\mathbf{Y}_n; \tilde \theta )$    \\ \hline
A    &  {$1.61*10^{-4}$} & $-2.11*10^{-5}$ & $-1.27*10^{-4}$  \\ 
B  & {$1.62*10^{-4}$} & -12.1 & -10.3          \\ \hline
\hline
\end{tabular}
\caption{likelihood-based Sub-Statistics for Policyholders A and B}
\label{tabexam}
\end{table}

Now consider the estimated Bayesian premiums under the exponential principle in Table \ref{premsexample}. The premiums are computed as the product of the manual premium and the rating factor given in the pricing formula. We observe that the resulting premiums align with our intuitive analysis. Specifically, policyholder B, who has a claim history with two claims, is expected to pay almost twice as much in premiums as policyholder A, who has no claims in the past four years. Therefore, the claim history is the main driver of the differences in premiums, and the likelihood-based statistic is the measure that captures such discrepancies. 

\begin{table}[!htbp]
\centering
\small
\begin{tabular}{cccc}
\hline
\hline
\small
\textbf{Client ID} & \textbf{Manual Premium}     & $\exp ( g( \mathcal{L} (\mathbf{Y}_n; \tilde{\theta} ),n))$ & \textbf{Bayesian Premium}  \\ \hline
A & 0.253 &  0.854  & 0.216 \\ 
B & 0.253  & 1.696  & 0.429  \\ \hline
\hline
\end{tabular}
\caption{Calculation of Bayesian  premiums}
\label{premsexample}
\end{table}

\section{Conclusions}
\label{conclusion}

Performing accurate experience rating on large insurance portfolios is a challenging task due to two major problems: 1) accounting for the heterogeneity of the policyholders requires flexible, and possibly not mathematically tractable models that can fit complexity in the risk behavior of policyholders and  2) it is necessary to have large computational power to deal with large sized insurance portfolios, especially when no analytical solutions are available. The first issue can be partially addressed with general Bayesian models beyond the simplistic assumptions commonly used in insurance ratemaking. However, these methods heavily rely on computational techniques which may be affected due to the second problem. Therefore, it is of great importance for actuaries to effectively address the computational issues and to have an effective ratemaking system that is transparent and suited to actuarial standards.

In this paper, we propose a methodology that addresses these challenges by computing the Bayesian premium through a surrogate modeling approach based on a tailored-made summary statistic that we term the likelihood-based statistic. It can be seen as a measure of how likely it is for a policyholder to experience a certain claim history, under a different probability measure, and it is a sufficient statistic for several distribution families. This approach provides an analytical expression for the Bayesian premium, i.e. a pricing formula, which can alleviate the expensive calculation of Bayesian premiums in large portfolios and enable the actuary to interpret the results. The actuary can rely on the pricing formula to approximate premiums for any policyholder for any parametric model. Additionally, this expression enables the actuary to provide a transparent picture of the ratemaking process to clients and regulators and provides a reliable way of performing risk classification among the policyholders.

Future research can explore the application of the surrogate model and the likelihood-based statistic in the context of evolutionary credibility. Indeed, more recent claims may provide a better assessment of the current risk behavior of a policyholder, and so should have a larger impact than older claims when upgrading the premiums. This is usually achieved in insurance with State-Space models, e.g., \cite{ahn2021ordering}, which can still be embedded into the Bayesian model in Section \ref{framework} under a high dimensional latent variable vector. Therefore the methodology still applies, but much more research must be performed in future case studies.

Another future research direction is the construction of other summary statistics that complement the likelihood-based statistic. The likelihood-based statistic may not be the only quantity on which the claim history affects the predictive distribution of policyholders. Therefore, it may be possible to improve the accuracy of the pricing formula by considering other types of likelihood-based statistics in the model. Another research direction is to consider the development of a likelihood-based statistic that is ``non-parametric." The likelihood-based statistic here proposed is model-dependent, as it relies on a given parametric model. Thus, it is desirable to develop a likelihood-based statistic based on empirical data sets.

Finally, it is worth mentioning that the surrogate modeling approach can find applications in other areas not directly related to experience rating. For instance, this could be applied to predictive models for claim reserving to account for the policyholders' experiences and the current claims on development. Similarly, the applications of the likelihood-based statistic can be extended to the generality of Bayesian inference. For instance, a surrogate model can be used in approximate Bayesian computations to approximate expectations, particularly in EM algorithms to make them more efficient.

\section*{Acknowledgment}
This work was partly supported by Natural Sciences and Engineering
Research Council of Canada [RGPIN 284246, RGPIN-2017-06684].  The authors would like to thank an anonymous referee for their suggestions that improved the manuscript considerably.

\bibliographystyle{apalike}
\small
\bibliography{references}

\appendix

\section{  On the sufficiency of the likelihood-based summary statistic }
\label{sufficiency}

The primary motivation behind introducing the {likelihood-based statistic is to uncover a summary statistic of a policyholder's claim history that leads to a similar posterior predictive distribution when it is used instead of the whole data, thus rendering it suitable for the surrogate modeling as explained in Section \ref{surrogate}.

In this section, we present two key results concerning this property. The first result, detailed in Corollary \ref{corl}, demonstrates that inference based on the likelihood-based statistic approximates inference based on the complete claim history of a policyholder when the values of $\tilde \theta$ are judiciously selected. The second result, as outlined in Proposition \ref{suffindex} and supported by subsequent examples, delves even deeper by revealing that a subtle simplification of the likelihood-based statistic can serve as a sufficient statistic for certain family distributions. Notably, this includes the exponential dispersion family of distributions.

\subsection{ Approximate inference }

Here, we derive a significant result, presented in Corollary \ref{corl}, which affirms that the posterior predictive distribution based on the recently introduced likelihood-based statistic closely approximates the true posterior predictive distribution based on the complete claim history. This approximation can be achieved to a high degree of accuracy by ensuring that the set of values $\tilde \theta$ are properly selected.

The result is both intuitive and relatively straightforward to demonstrate. As mentioned earlier, the likelihood function (or log-likelihood) is always a sufficient statistic regardless of the model (Lemma 1 in \cite{mayoqualitative} or \cite{schweder2016confidence}). In mathematical terms, this implies that the set of statistics given by $\{ \mathcal{L} (\mathbf{Y}_n; \tilde \theta) \}_{\tilde \theta \in \mathcal{R}_{\Theta}}$,  serves as a sufficient statistic for $\Theta$, and therefore $P(Y_{n+1} \vert\{ \mathcal{L} (\mathbf{Y}_n; \tilde \theta) \}_{\tilde \theta \in \mathcal{R}_{\Theta}}) = P(Y_{n+1} \vert \mathbf{Y}_n)$. Here recall that the likelihood-based statistic is defined as the log-likelihood evaluated at a fixed value $\tilde \theta$, and therefore equivalent to the whole log-likelihood function when considering all the values of the parameter space.

Our likelihood-based statistic, however, does not encompass all the possibilities within the parameter space. Instead, it focuses on a subset of it, say $\mathcal{I} \subset \mathcal{R}_{\Theta} $,  and the inference is based solely on the associated subset of summary statistics $\{ \mathcal{L} (\mathbf{Y}_n; \tilde \theta) \}_{ \tilde \theta \in \mathcal{I}}$. Proposition \ref{limitsuff} below demonstrates that as we allow the set $\mathcal{I}$ to expand to cover the entire parameter space $\mathcal{R}_{\Theta}$, the inference approaches the true predictive distribution more closely.

\begin{proposition}
Let $ \mathcal{I}_1 \subset  \mathcal{I}_2 \subset \ldots \subset   \mathcal{I}_m \subset \ldots \subset \mathcal{R}_{\Theta} $ be a sequence of sets of values $\tilde \theta \in \mathcal{R}_{\Theta}$ that  increases to $\mathcal{R}_{\Theta} $. Let   $\{ \mathcal{L} (\mathbf{Y}_n; \tilde \theta) \}_{ \tilde \theta \in \mathcal{I}_m}$ be the  likelihood-based  statistics associated to the respective values $\tilde \theta \in \mathcal{I}_m $. Then,
$$
P(Y_{n+1} \in A \vert \{ \mathcal{L} (\mathbf{Y}_n; \tilde \theta) \}_{ \tilde \theta \in \mathcal{I}_m} )  \xrightarrow{\: m \to \infty \: } P(Y_{n+1} \in A \vert \mathbf{Y}_n)
$$
for every measurable set $A$.
\label{limitsuff}
\end{proposition}
\begin{proof}
Define the sequence of sigma-algebras $\mathcal{F}_m = \sigma( \{\mathcal{L} (\mathbf{Y}_n; \tilde \theta) \}_{ \tilde \theta \in \mathcal{I}_m})$ and $\mathcal{F}_\infty = \sigma( \{\mathcal{L} (\mathbf{Y}_n; \tilde \theta) \}_{ \tilde \theta \in \mathcal{R}_{\Theta}})$. As the sequence of sets $\mathcal{I}_m 	\uparrow \mathcal{R}_{\Theta}$, then we also have $\mathcal{F}_m \uparrow \mathcal{F}_{\infty}$. Therefore by, theorem 4.6.8 of \cite{durrett2019probability} applied to the indicator variable of the event $A$  we have:
$$
P(Y_{n+1} \in A \vert \{ \mathcal{L} (\mathbf{Y}_n; \tilde \theta) \}_{ \tilde \theta \in \mathcal{I}_m} )  = E( \mathbf{1}_{ \{Y_{n+1} \in A \}} \vert \mathcal{F}_m ) \xrightarrow{\: m \to \infty \: } E( \mathbf{1}_{ \{Y_{n+1} \in A \} } \vert \mathcal{F}_{\infty} ) =  P(Y_{n+1} \in A \vert \{ \mathcal{L} (\mathbf{Y}_n; \tilde \theta) \}_{ \tilde \theta \in \mathcal{R}_{\Theta}} )
$$
Lastly, the result follows by noting that $P(Y_{n+1} \in A \vert\{ \mathcal{L} (\mathbf{Y}_n; \tilde \theta) \}_{\tilde \theta \in \mathcal{R}_{\Theta}}) = P(Y_{n+1} \in A \vert \mathbf{Y}_n)$ because of the sufficiency of the likelihood function.
\end{proof}

\begin{corollary}
\label{corl}
Let $\epsilon >0$ and $A$ be a measurable set. Then there exists a proper subset $\mathcal{I} \subset \mathcal{R}_{\Theta} $ for which the associated likelihood-based  statistics satisfy:
$$
\vert    P(Y_{n+1} \in A \vert \{ \mathcal{L} (\mathbf{Y}_n; \tilde \theta) \}_{ \tilde \theta \in \mathcal{I}} )  - P(Y_{n+1} \in A \vert \mathbf{Y}_n)  \vert < \epsilon
$$
\end{corollary}
\begin{proof}
This follows immediately from the definition of limit.
\end{proof}

The preceding propositions emphasize that the proper selection of $\tilde \theta$ leads to effective performance in inference tasks based on the likelihood-based summary statistic. However, it is important to note that this result is more conducive to a conceptual understanding than a practical application, as the proof does not offer a constructive method for selecting the values of $\tilde \theta$. Furthermore, in our selection, we utilize a single value of $\theta$ per policyholder, whereas the approximation in the previous corollary may necessitate a more diverse selection to enhance performance.

While increasing the number of likelihood-based summary statistics per policyholder (and thus increasing the number of values of $\tilde \theta$ used) could potentially achieve a better performance, it is worth noting that this may not always be necessary. For example, Proposition \ref{suffindex} below demonstrates that a single $\tilde \theta$ suffices to obtain a sufficient statistic for the exponential family, thereby enabling exact inference. Similarly, our simulation study in Section \ref{simulation} indicates that, empirically, the approximation based on a single value of $\tilde \theta$ can yield satisfactory performance.

\subsection{A condition for sufficiency and examples}

To show the results on sufficiency, we first observe that the likelihood-based statistic can be simplified without any consequences. Note that if we can decompose  $\ell( \mathbf{Y}_n \vert \Theta) = l_1(\mathbf{Y}_n) + l_2(\mathbf{Y}_n,\Theta)$ for some functions $l_1, l_2$ (not necessarily probability functions), then we can omit the first term that depends only on $\mathbf{Y}_n$, as it simplifies both in the numerator and denominator of the predictive distribution in Equation (\ref{eqn1}). Therefore we can potentially find a simplified form of the likelihood-based statistic without losing any information when performing inference on the posterior predictive distribution. With this in mind, we introduce a refined version of the likelihood-based statistic and some of its properties.

\begin{definition}
The refined likelihood-based statistic denoted as $ \tilde{\mathcal{L}}(\mathbf{Y}_n; \tilde{\theta} ) $, is  defined as the ``minimal" term in the following additive  decomposition of the likelihood-based statistic:

$$
\mathcal{L} (\mathbf{Y}_n; \tilde{\theta} ) = L(\mathbf{Y}_n)+\tilde{\mathcal{L}} (\mathbf{Y}_n; \tilde{\theta} ).
$$

The term ``minimal" is defined in the sense that if there is a further sub-decomposition of the form $\tilde{\mathcal{L}} (\mathbf{Y}_n; \tilde{\theta} ) = L_1(\mathbf{Y}_n)+L_2 (\mathbf{Y}_n; \tilde{\theta} )$ for some functions $L_1, L_2$, then we must have $ \tilde{\mathcal{L}} (\mathbf{Y}_n; \tilde{\theta} ) =  L_2 (\mathbf{Y}_n; \tilde{\theta} ) $. 
\end{definition}

We now show that the refined likelihood-based statistic may be a sufficient statistic for the latent variable $\Theta$ under certain conditions, and so guarantee the effectiveness of the likelihood-based statistic in capturing the information of the claim history. To do so, let $T(\mathbf{Y}_n)$ be a sufficient statistic for the latent variable $\Theta$. By the Fisher–Neyman factorization theorem, e.g.,  \cite{casella2021statistical}, the conditional likelihood function can be factored into two non-negative functions, thus the log-likelihood function can be additively separated into two functions $l_1$ and $l_2$ as:
$$
\ell (\mathbf{Y}_n \vert \Theta) = l_1(\mathbf{Y}_n) + l_2(T(\mathbf{Y}_n),\Theta) ~~ \forall \mathbf{Y}_n, \Theta.
$$

\begin{proposition}
Suppose there exists a value $\tilde \theta \in R_{\Theta}$ at which $l_2(T(\mathbf{Y}_n),\tilde \theta)$ is a one-to-one function when viewed as a function of the sufficient statistic $T(\mathbf{Y}_n)$. Then the refined likelihood-based statistic $\tilde{\mathcal{L}} (\mathbf{Y}_n; \tilde \theta)$  at the value $\tilde \theta$, is a sufficient statistic for the latent variable $\Theta$.
\label{suffindex}
\end{proposition}

\begin{proof}
The decomposition of the log-likelihood above holds for every value of $\theta \in \mathcal{R}_{\Theta}$, including the one in the assumption. When fixing the value of $\Theta$, the log-likelihood function at the left-hand side of the decomposition  becomes the likelihood-based statistic, and so at this particular value $\tilde \theta$  we must have the following relationship in terms of the refined likelihood-based statistic:

$$
\tilde{\mathcal{L}} (\mathbf{Y}_n; \tilde \theta) + L(\mathbf{Y}_n)= l_1(\mathbf{Y}_n) + l_2(T(\mathbf{Y}_n), \tilde \theta) 
$$

And so we  have:

$$
\tilde{\mathcal{L}} (\mathbf{Y}_n; \tilde \theta) = ( l_1(\mathbf{Y}_n) - L(\mathbf{Y}_n)) + l_2(T(\mathbf{Y}_n), \tilde \theta) 
$$

The last expression provides an additive decomposition of the refined likelihood-based statistic, and by the construction of its minimality, we must have  that:

$$
\tilde{\mathcal{L}} (\mathbf{Y}_n; \tilde \theta) = l_2(T(\mathbf{Y}_n), \tilde \theta) 
$$

Now, by the assumption of the function $l_2(T(\mathbf{Y}_n), \tilde \theta)$ at that specific value $\tilde \theta$, the right-hand side is a one-to-one function when viewed as a function of  $T(\mathbf{Y}_n)$. Therefore, the left-hand side is a one-to-one function of a sufficient statistic and so the refined likelihood-based statistic $\tilde{\mathcal{L}}  (\mathbf{Y}_n; \tilde \theta)$ is also a sufficient statistic for the latent variable $\Theta$. Moreover, if $T(\mathbf{Y}_n)$ is also minimal sufficient, then by the one-to-one correspondence, $\tilde{\mathcal{L}} (\mathbf{Y}_n, \tilde \theta)$ is also minimal sufficient. 
\end{proof}

\begin{example}[Exponential dispersion family]
In the particular case in which the model distribution $f(Y \vert \Theta, \mathcal{O})$ is given by a member of the exponential dispersion family of distributions with a dispersion parameter that does not depend on the latent variable. This is a construction widely used for insurance applications, e.g., \cite{wuthrich2022statistical}, we have:
$$
f(Y_j \vert \Theta = \theta, \mathcal{O})  = \exp( \frac{\theta S(Y_j) -C(\theta)}{\varphi_j} -Q(Y_j;\varphi_j))
$$

for some functions $S(\cdot), C(\cdot), Q(\cdot)$ and the set of parameters $\mathcal{O}$ contains only the dispersion parameters $\varphi_j$. That said, the conditional log-likelihood takes the form:
$$
\ell (\mathbf{Y}_n \vert \theta) =  \sum_{j=1}^n \frac{\theta S(Y_j) -C(\theta)}{\varphi_j} -\sum_{j=1}^n Q(Y_j;\varphi_j)
$$

Thus the  likelihood-based statistic and the refined likelihood-based statistic are respectively:
$$
\mathcal{L} (\mathbf{Y}_n; \tilde \theta) =  \tilde \theta \sum_{j=1}^n \frac{ S(Y_j)}{\varphi_j}   - \sum_{j=1}^n \frac{C(\tilde \theta)}{\varphi_j} -\sum_{j=1}^n Q(Y_j;\varphi_j)
$$
$$
\tilde{\mathcal{L}}  (\mathbf{Y}_n; \tilde \theta) = \tilde \theta \sum_{j=1}^n \frac{ S(Y_j)}{\varphi_j}   - \sum_{j=1}^n \frac{C(\tilde \theta)}{\varphi_j}
$$

It is known that $\sum_{j=1}^n \frac{ S(Y_j)}{\varphi_j}$ is the minimal sufficient statistic for the exponential dispersion family. Observe that the function linking the sufficient statistic and the refined likelihood-based statistic is a one-to-one function for $\tilde \theta \neq 0$. Therefore, by the proposition above, the refined likelihood-based statistic is also a minimal sufficient statistic. 
\end{example}

We note that the likelihood-based statistic is one-dimensional as it reduces the information of the whole claim history of a policyholder to a single quantity. Therefore, the likelihood-based statistic can potentially be a sufficient statistic when the latent variable $\Theta$ is one-dimensional. In the case of multivariate models for $Y_j$ with $D$ dimensions, this limitation can be weakened in the sense that now the likelihood-based sub-statistics can also play a role to constitute a set of summary statistics that all together may be sufficient for a latent variable $\Theta$. Indeed, the result of the proposition above extends directly if we consider a multivariate model for $Y_j$ in which each of the components $Y_j^{(d)}$  is associated with at most one of the components of the latent variable $\tilde \Theta ^{(d)}$ at a time. We illustrate this generalization with just an example as follows:

\begin{example}[Multivariate exponential dispersion family ]
Consider now the case of a particular multivariate exponential dispersion family of distributions with dispersion parameters that do not depend on the latent variable. This construction is also widely used for applications in insurance.
$$
f( Y_j \vert \Theta = \theta, \mathcal{O}) = \exp( \sum_{d=1}^D \frac{ \theta^{(d)} S^{(d)}(Y_j^{(d)}) -   C^{(d)}(\theta^{(d)})}{\varphi_j^{(d)}} -  \sum_{d=1}^D  Q^{(d)}(Y_j^{(d)}; \varphi_j^{(d)}) )
$$

for some functions $S^{(d)}(\cdot), C^{(d)}(\cdot), Q^{(d)}(\cdot)$. Again, the set of parameters $\mathcal{O}$ is associated with the dispersion parameters. Observe that in this particular case the function $C(\theta)$ is split additively into each of the components, which is the case when the components of the vector $Y_j$ are conditionally independent given $\Theta$. The conditional log-likelihood takes the form
$$
\ell (\mathbf{Y}_n \vert \theta) = \sum_{j=1}^n  \sum_{d=1}^D \frac{ \theta^{(d)} S^{(d)}(Y_j^{(d)}) -   C^{(d)}(\theta^{(d)})}{\varphi_j^{(d)}} - \sum_{j=1}^n  \sum_{d=1}^D  Q^{(d)}(Y_j^{(d)}; \varphi_j^{(d)}), 
$$

And so the  likelihood-based statistic and sub-statistics are respectively
$$
\mathcal{L} (\mathbf{Y}_n; \tilde \theta) = \sum_{d=1}^D  \tilde \theta^{(d)}  \sum_{j=1}^n   \frac{  S^{(d)}(Y_j^{(d)})}{\varphi_j^{(d)}} -  \sum_{d=1}^D \sum_{j=1}^n  \frac{  C^{(d)}(\tilde \theta^{(d)})}{\varphi_j^{(d)}} -  \sum_{d=1}^D \sum_{j=1}^n  Q^{(d)}(Y_j^{(d)}; \varphi_j^{(d)}), 
$$

$$
\mathcal{L}^{(d)} (\mathbf{Y}_n^{(d)}; \tilde \theta^{(d)}) =  \tilde \theta^{(d)}  \sum_{j=1}^n   \frac{  S^{(d)}(Y_j^{(d)})}{\varphi_j^{(d)}} -  \sum_{j=1}^n  \frac{  C^{(d)}(\tilde \theta^{(d)})}{\varphi_j^{(d)}} - \sum_{j=1}^n  Q^{(d)}(Y_j^{(d)}; \varphi_j^{(d)}), 
$$

$$
\mathcal{L} (\mathbf{Y}_n; \tilde \theta) = \sum_{d=1}^D  \mathcal{L}^{(d)} (\mathbf{Y}_n^{(d)}; \tilde \theta^{(d)}).
$$

The refined likelihood-based statistic and refined sub-statistics are respectively

$$
\tilde{\mathcal{L}} (\mathbf{Y}_n; \tilde \theta) = \sum_{d=1}^D  \tilde \theta^{(d)}  \sum_{j=1}^n   \frac{  S^{(d)}(Y_j^{(d)})}{\varphi_j^{(d)}} -  \sum_{d=1}^D \sum_{j=1}^n  \frac{  C^{(d)}(\tilde \theta^{(d)})}{\varphi_j^{(d)}},
$$

$$
\tilde{\mathcal{L}} ^{(d)} (\mathbf{Y}_n^{(d)}; \tilde \theta^{(d)}) =\tilde \theta^{(d)}  \sum_{j=1}^n   \frac{  S^{(d)}(Y_j^{(d)})}{\varphi_j^{(d)}} -  \sum_{j=1}^n  \frac{  C^{(d)}(\tilde \theta^{(d)})}{\varphi_j^{(d)}},
$$

$$
\tilde{\mathcal{L}}  (\mathbf{Y}_n; \tilde \theta) = \sum_{d=1}^D  \tilde{\mathcal{L}} ^{(d)} (\mathbf{Y}_n^{(d)}; \tilde \theta^{(d)}).
$$

It is known that the set of statistics $\sum_{j=1}^n   \frac{  S^{(d)}(Y_j^{(d)})}{\varphi_j^{(d)}}, ~ d =1, \ldots, D $ are the minimal sufficient statistic for this family,  and the function linking each of these sufficient statistics and the associated refined likelihood-based sub-statistics are one-to-one functions for $\tilde \theta^{(d)} \neq 0$. Therefore, the set of likelihood-based sub-statistics is also a minimal sufficient statistic. 
\end{example}

\end{document}